\documentclass[journal,twoside,web]{ieeecolor}
\usepackage{generic}
\usepackage{amsmath,amssymb,amsfonts,bm,bbm,scalerel,mathtools}

\usepackage{amsthm}

\usepackage{graphicx}
\usepackage{textcomp}
\usepackage{hyperref}
\usepackage{float}
\usepackage{comment}
\usepackage{tikz}
\usepackage{caption}
\usepackage{subcaption}
\allowdisplaybreaks
\usepackage{cite}

\def\BibTeX{{\rm B\kern-.05em{\sc i\kern-.025em b}\kern-.08em
    T\kern-.1667em\lower.7ex\hbox{E}\kern-.125emX}}
\theoremstyle{plain}
\newtheorem{thm}{Theorem}
\newtheorem{lem}[thm]{Lemma}
\newtheorem{prop}[thm]{Proposition}
\theoremstyle{definition}
\newtheorem{defn}[thm]{Definition}
\newtheorem{problem}[thm]{Problem}

\DeclareMathAlphabet{\mymathbb}{U}{BOONDOX-ds}{m}{n}

\newcommand{\R}{\mathbb{R}}

\newcommand{\W}{\mathcal{W}}
\newcommand{\K}{\mathcal{K}}

\newcommand{\D}{\mathfrak{D}}

\newcommand{\half}{\tfrac{1}{2}}

\newcommand{\lp}{\left(}
\newcommand{\rp}{\right)}

\newcommand{\of}{\circ}
\newcommand{\QDbh}{\hat{\bar{Q}}_{\rmd}}

	\newcommand{\req}[1]{(\ref{#1.eq})} 
	\newcommand{\be}{\begin{equation}} 
	\newcommand{\ee}{\end{equation}} 
	\newcommand{\bbm}{\begin{bmatrix}}
	\newcommand{\ebm}{\end{bmatrix}}

	\newcommand{\rmr}{{\rm r}}
	\newcommand{\rmd}{{\rm d}} 
	 
	\newcommand{\rmx}{{\rm x}} 
	\newcommand{\rmy}{{\rm y}} 
	\newcommand{\rmm}{{\rm m}} 
	\newcommand{\rma}{{\rm a}} 
	\newcommand{\hstm}{\hspace{2em}}
	\newcommand{\hsom}{\hspace{1em}}
	\newcommand{\rom}{{\rule{0em}{1em}}}
	\newcommand{\romn}{\rule{0em}{.9em}}

        \newcommand{\sm}{{\text{-}}}
        \newcommand{\marg}{{{\Pi}}}

        \newcommand{\ssR}{{\scriptscriptstyle R}}
        \newcommand{\ssD}{{\scriptscriptstyle D}}
        \newcommand{\QR}{{Q_{\!\ssR}}}
        \newcommand{\QRh}{{\hat{Q}_{\!\ssR}}}
        \newcommand{\QD}{{Q_{\!\ssD}}}
        \newcommand{\QDb}{{\bar{Q}_{\!\ssD}}}
        
        \newcommand{\FR}{{F_{\!\ssR}}}

        \newcommand{\lcb}{\left\{}
        \newcommand{\rcb}{\right\}}
        \newcommand{\lB}{\!\left(}
        \newcommand{\rB}{\right)}
        \newcommand{\sss}{\scriptscriptstyle}
        
        \newcommand{\bone}{{\mathbbm{1}}}
        \newcommand{\smint}[2]{{\scaleobj{.8}{\int_{{#1}}^{{#2}}}}}
        \newcommand{\smsum}[2]{{\scaleobj{.8}{\sum_{{#1}}^{{#2}}}}}
        \newcommand{\caP}{{\cal P}}
        \newcommand{\caI}{{\cal I}}
        
        \newcommand{\rmq}{{\rm q}}
        \newcommand{\Pfav}[1]{\bar{#1}}
        \newcommand{\Pav}[1]{\bar{#1}}
        
        \newcommand{\Pq}{{\cal P}_{\!\rmq}}

        \newcommand{\ric}{{\sf p}}
        \newcommand{\sfY}{{\sf Y}}
        \newcommand{\sfM}{{\sf M}}
        \newcommand{\ff}{{\sf y}}
        \newcommand{\sT}{{\scriptstyle T}}

        \newcommand{\rha}{\hat{r}}
        \newcommand{\dha}{\hat{d}}
        \newcommand{\mycaption}[1]{\caption{\footnotesize #1}}

\newcommand\copyrighttext{%
	\footnotesize\textcopyright 2025 IEEE. Personal use of this material is permitted.  Permission from IEEE must be obtained for all other uses, in any current or future media, including reprinting/republishing this material for advertising or promotional purposes, creating new collective works, for resale or redistribution to servers or lists, or reuse of any copyrighted component of this work in other works.}
\newcommand\copyrightnotice{
	\begin{tikzpicture}[remember picture,overlay]
		\node[anchor=south,yshift=3pt] at (current page.south) {\fbox{\parbox{\dimexpr\textwidth-\fboxsep-\fboxrule\relax}{\copyrighttext}}};
\end{tikzpicture}}
        
\begin{document}

\title{Optimal Assignment and Motion Control in Two-Class Continuum Swarms} 

\author{Max Emerick, \IEEEmembership{Graduate Student Member, IEEE}, Stacy Patterson, and Bassam Bamieh, \IEEEmembership{Fellow, IEEE}
\thanks{This article was submitted for review on July 24, 2024. This work was supported in part by NSF award ECCS-1932777 and CMI-1763064.}
\thanks{Max Emerick and Bassam Bamieh are with the Dept. of Mechanical Engineering, University of California Santa Barbara ({\em \{memerick,bamieh\}@ucsb.edu}). Stacy Patterson is with the Dept. of Computer Science, Rensselaer Polytechnic Institute ({\em sep@cs.rpi.edu}).}
\thanks{A portion of the results in this paper appeared in a preliminary version in the proceedings of the 9th IFAC Conference on Networked Systems (NecSys22) \cite{Emerick2022}. That conference article was a preliminary announcement of the results and did not include new mathematical insights and proofs presented here, nor the examples and the important case of periodically time-varying demands.}
\thanks{The authors thank Connor Hughes for valuable discussions in the early stages of this work.}}

\maketitle

\begin{abstract}                
We consider optimal swarm control problems where two different classes of agents are present. Continuum idealizations of large-scale swarms are used where the dynamics describe the evolution of the spatially-distributed densities of each agent class. The problem formulation we adopt is motivated by applications where agents of one class are assigned to agents of the other class, which we refer to as demand and resource agents respectively. Assignments have costs related to the distances between mutually assigned agents, and the overall cost of an assignment is quantified by a Wasserstein distance between the densities of the two agent classes. When agents can move, the assignment cost can decrease at the expense of a physical motion cost, and this tradeoff sets up a nonlinear infinite-dimensional optimal control problem. We show that in one spatial dimension, this problem can be converted to an infinite-dimensional, but decoupled, linear-quadratic (LQ) tracking problem when expressed in terms of the quantile functions of the respective agent densities. Solutions are given in the general one-dimensional case, as well as in the special cases of constant and periodically time-varying demands.
\end{abstract}


\copyrightnotice
\vspace{-5mm}


\section{Introduction}

Low-cost sensing, processing, and communication hardware is driving the use of autonomous swarms of robotic agents in diverse settings, including emergency response, transportation, logistics, data collection, and defense. Large swarms can have significant advantages in efficiency and robustness. However, as swarms scale in size it becomes increasingly difficult to plan and coordinate motion between agents. For sufficiently large swarms, modeling the swarm as a density distribution over the domain (i.e. as a continuum) provides a significant model reduction as well as improved insight into the macroscopic behavior of the swarm. Thus, the development of motion planning and control strategies for systems described by distributions is a problem of interest.

The problem formulation we propose in this paper is partially motivated by applications in edge computing \cite{He2021,Yang2019,Zhou2018} and mobile cloudlets \cite{Li2014,Chen2015}, as will be explained in the next section on problem formulation. We develop a model where densities of two classes of agents -- referred to as demand and resource -- interact in a physical space. We propose a problem where demands and resources must be dynamically matched/assigned, while resources are physically redistributed to lower the cost of assignment. In this way, our problem comprises two parts: spatio-temporal dynamic matching \cite{Ayala2018,Lowalekar2018,Kanoria2021} and spatio-temporal control.

We approach this problem using tools from the areas of optimal mass transport and optimal control. While our particular problem formulation is new, there has been a recent surge of interest in the connections between these two fields. The first contact here was made in the classic paper \cite{Benamou2000}, where it was shown that the optimal transport problem has a ``dynamic'' formulation as an optimal control problem. More recently, there have been numerous papers investigating problems of optimal transport and optimal control for multi-agent swarms in the contexts of self-interaction \cite{Fornasier2014,Carrillo2014,Bonnet2019,Bonnet2019a,Bonnet2021,Bongini2017,Jimenez2020,Burger2021,Chen2023}, coverage control \cite{Bandyopadhyay2014,Krishnan2018a,Inoue2021,Haasler2021,Terpin2024}, shape control \cite{Lin2022}, and tracking \cite{Kachar2022,Abdulghafoor2023}. The work we present here is most similar to these last two references in that we focus on the problem of tracking. However, it differs in that we investigate a different model, and obtain mainly analytical (rather than numerical) results.

We  point out that while distributed decision making approaches are important in multi-agent systems, we adopt a centralized control approach in the current work and do not consider the problem of distributed implementation. We do this to understand ultimate performance limitations, i.e., to provide benchmarks by which to evaluate  control strategies. This idea is motivated by the use of optimal control problems for co-design of systems and controllers, rather than designing controllers for an existing system.

The rest of the paper proceeds as follows. We formulate (Section~\ref{prob_formulation}) an optimal control problem where the cost is a tradeoff between the cost of assignments (quantified by the Wasserstein distance), and the physical cost of moving resources. The decision variable is the velocity field of the resource agents, which follow an advection model. As such, this is an infinite-dimensional nonlinear optimal control problem. In the special case of one spatial dimension (Section~\ref{1D}) we recast the problem in terms of the quantile functions of the densities, and show that, remarkably, this problem transforms into a linear-quadratic (LQ) tracking problem. The problem is still infinite dimensional, but with a largely uncoupled structure which allows for explicit solutions. These solutions are then re-expressed in terms of the original problem variables. We apply these solutions in the two special cases (Section~\ref{explicit.sec}) of constant and periodically time-varying demand to demonstrate the method. In the case of constant demand, we find that resource agents follow Wasserstein geodesics, but with a time schedule determined by the solution of the optimal control problem. In the case of periodic demand, solutions are given in terms of temporal Fourier transforms, and optimal motions are interpreted as filtered versions of tracking signals. We conclude (Section~\ref{conclusion}) with a brief discussion and pointers to the many possible future directions.

\subsection*{Notation and Preliminaries} \label{preliminaries}

For an introduction to optimal transport and the Wasserstein distance, see \cite{Santambrogio2010}.
We use generalized functions to describe  densities, i.e.,  Dirac distributions rather than ``atomic measures''. For a formal treatment of generalized functions, see \cite{Gelfand1964}. Notations like $R_t(x)$ = $R(x,t)$ or $R_t$ = $R_t(\cdot)$ = $R(\cdot,t)$ are used interchangeably to emphasize $R$ either as a parameterized curve in function space or as a spatio-temporal field.


\section{Problem Formulation} \label{prob_formulation}

We consider settings where two classes of mobile agents -- referred to as {\em demand} and {\em resource} -- interact in a physical space. We consider macroscopic descriptions of both agent classes so that each can be modeled as a density function on the domain. Physical space is modeled as a bounded convex subset $\Omega \subset \R^n$, where $n = 1,~ 2,~ \text{or} ~ 3$.

Demand agents are assumed to be lightweight (for example, having little processing power or memory, but agile for use in sensing and exploration). They can communicate and offload their computing and long-distance communication tasks to heavier-duty ``resource'' agents, which tend to be less agile. For example, demand agents may be  inexpensive camera-equipped drones that are tasked with surveilling a disaster area to assist in search and rescue~\cite{Lyu2023}. These drones have limited space and limited battery capacity, and thus cannot provide complex on-board video processing.  Instead, these drones offload  video analysis tasks to higher resourced mobile edge computing servers, the resource agents, which can be larger fixed-wing drones equipped with sufficient computing power~\cite{Adnan2024,Blair2022}.

We define a (possibly time-varying) {\em assignment} map between demand and resource, encoding which tasks are assigned to which agents. Assignment maps incur a cost determined by the physical distance between paired agents, and the minimum cost of assignment is seen to be the Wasserstein distance between the resource and demand densities. This cost can capture a latency penalty in the surveillance application, for example, due to transmitting video data from demand agents to resource agents. We assume that demand agents' physical locations are primarily determined by the tasks assigned to them, and so is an external signal in our problem, whereas resource agents' motions are the decision variables. The assignment cost can thus decrease if resources move closer to demands, but this has an additional physical {\em motion} cost for the system. Thus assignment and motion costs are competing objectives, and we formulate a control problem that provides the optimal tradeoff between these two costs.

Our mathematical model thus comprises the following five components:

\textbf{(1)} 
A \textit{demand distribution} $D_t(\cdot)$ defined over $\Omega\subset\R^n$, which in its most general form is 
\be
D_t(x) ~=~ \rmd(x,t) ~+~ \smsum{k=1}{N_d} \rmd_k(t) ~\delta \left( x-\xi_k(t) \right) .
\label{Dt_def.eq}
\ee
The demand distribution thus has a continuous and a discrete component: $\xi_k$ represents the spatial location of the $k$'th discrete agent whose demand is quantified by the (possibly time-varying) function $\rmd_k(\cdot)$, while the continuous part $\rmd(\cdot,\cdot)$ represents a continuum model of a large-scale system that is best described in terms of its demand density.

\textbf{(2)} 
A \textit{resource distribution} $R_t(\cdot)$ also defined over $\Omega \subset \R^n$, describing the distribution of resources at time $t$
\be
R_t(x) ~=~ \rmr(x,t) ~+~ \smsum{k=1}{N_\rmr} \rmr_k(t) ~\delta \left( x-\eta_k(t) \right) , 
\label{Rt_def.eq}
\ee
with similar interpretations as those given to $D_t$ in~\req{Dt_def}.

In this paper, we assume that both distributions are nonnegative and are normalized to integrate to 1 at all times
\begin{equation} \label{normalized_mass_eq}
	\smint{\Omega}{} R_t(x)\,dx ~=~ \smint{\Omega}{} D_t(x)\,dx ~=~ 1 .
\end{equation}
$R_t$ and $D_t$ are thus elements of {\em the set of normalized nonnegative distributions over $\Omega$}, which we denote $\D(\Omega)$. (See the conclusion for comments on generalizing this assumption.)

\textbf{(3)}
An \textit{assignment plan} $\K_t (\cdot,\cdot)$, which is a normalized nonnegative distribution over $\Omega \times \Omega$. The value $\K_t(x,y)$ specifies the quantity of tasks assigned to resources at location $x$ by demands at location $y$ at time $t$. Thus, in the most general case, resource agents can handle tasks from multiple demand agents and vice versa.

With $D_t$ and $R_t$ normalized, the interpretation of the two-variable function $\K_t(\cdot,\cdot)$ as assignment of demands to resources gives the following ``marginalization property''\footnote{Since $R_t$ and $D_t$ are density functions with mass 1, they have interpretations as probability density functions of random variables. In this case, $\K_t$ can be interpreted as a joint distribution of those random variables.}
\be
\begin{aligned}
	R_t(x) ~&=~ \smint{\Omega}{} \K_t(x,y)\,dy ~=:~  \left( \marg_\rmx \K_t \right) (x) , \\
	D_t(y) 		~&=~  \smint{\Omega}{} \K_t(x,y)\,dx ~=:~\left(  \marg_\rmy \K_t \right) (y) .
\end{aligned}
\label{marg_const.eq}
\ee
We refer to $\marg_\rmx$ and $\marg_\rmy$ as the marginalization operators onto $x$ and $y$ respectively and write the above equations 
compactly using the notation $\marg \K_t ~=~ (R_t,D_t)$. The assignment plan $\K_t$ is formally an element of $\D(\Omega \times \Omega)$ with additional constraints given by \eqref{marg_const.eq}. The assignment plan is one of the decision variables in our optimal control problem. Each assignment plan incurs an assignment cost which will be defined shortly.

\textbf{(4)} The {\em spatio-temporal resource dynamics} are described by the continuity PDE
\begin{equation}
	\partial_t R_t(x) ~=~  -\nabla \cdot \big( V_t(x) \, R_t(x) \big) , 
	\label{advec.eq}
\end{equation}
where $\nabla \cdot := \sum_i \partial_{x_i}$ denotes the divergence operator, and $V_t(x):= \left[ v_1(x,t)  \cdots  v_n(x,t) \right]^T$ is a time-varying velocity field which ``steers'' resources in space. This  velocity field is also a decision variable in our optimal control problem. 

The continuity equation~\req{advec} is used for both continuum and discrete models (requiring, of course, the proper notion of weak solution).
Note that this velocity field is defined as a function of space, thus the velocity of each agent is determined by its location in space rather than its identity. Importantly, this leads to  an {\em implicit constraint} that any two agents at the same location must move with the same velocity.

\textbf{(5)} The {\em performance objective}. We first consider the cost associated with an assignment plan $\K_t$. This cost is quantified by the {\em total weighted (squared) distance}
\begin{equation}
	C_\rma (\K_t) ~:=~ \smint{\Omega \times \Omega}{} ||y-x||_2^2  ~\K_t(x,y) ~dx\,dy.
	\label{Ca_def.eq}
\end{equation}
The {\em assignment cost} \req{Ca_def} is ``location aware'' in the sense that assigning resources to demands that are far away incurs a high cost and vice versa. In our example, \req{Ca_def} is interpreted as the latency penalty for a given assignment.

The expression~\req{Ca_def} is familiar from optimal transport theory. 
At a fixed moment in time, optimization of $C_\rma$ subject to the 
marginalization constraints $\marg \K_t=(R_t,D_t)$ is known~\cite{Santambrogio2015} to be the {\em 2-Wasserstein distance} between $R_t$ and $D_t$
\be
\W^2_2(R_t,D_t) = 
\inf_{ \marg \K_t = (R_t,D_t) } 
\smint{\Omega\times\Omega}{} \|y-x\|_2^2 \, \K_t (x,y) \, dx\,dy,
\label{Kb_opt.eq}
\ee
while the optimal solution $\K_t$ is the {\em optimal transport plan}. We emphasize that in the present context, however, $C_\rma$ does not represent a physical transport cost (as in optimal transport theory), but rather an {\em assignment cost}, that is, the total cost of assignments of tasks from demands to resources. We therefore refer to an optimizer of~\req{Kb_opt} as an optimal {\em assignment plan} rather than an optimal {\em transport plan} to distinguish it from its traditional interpretation. In our context, the physical cost of motion is quantified differently, as we now explain.

It is clear that when $R_t$ and $D_t$ are far apart, the optimal assignment cost $\W^2_2(R_t,D_t)$ is high. It can be reduced if the resource $R_t$ can be ``transported'' (via~\req{advec}) to be closer to the demand $D_t$, but this must also have a cost. We quantify this cost of physical motion with the following {\em motion cost} 
\begin{equation} \label{v_cost.eq}
	C_\rmm \left( R_t , V_t \right) ~:=~ \smint{\Omega}{} \left\| V_t(x) \right\|_2^2 R_t(x) ~dx . 
\end{equation}
In our example, this quantity is interpreted as the instantaneous power required to overcome drag on the resource swarm. Over a finite-time maneuver, the time-integral of the above cost is thus the total energy expended in moving the resource agents.

Finally, we combine the two costs in the following aggregate expression for cost of maneuvers over a time horizon $[0,T]$ 
\begin{equation}
	J(\K,R,V) ~:=~   \smint{0}{T} \lp C_\rma \left( \K_t \right) +\alpha^2 C_\rmm \left( R_t,V_t \right) \rp dt , 
	\label{J_def.eq}
\end{equation}
where $\alpha>0$ is a ``trade off'' parameter. The two objectives $C_\rma$ and $C_\rmm$ are clearly competing. If motion cost were negligible ($\alpha \ll 1$), then the optimal solution would be to move $R_t$ quickly so that it matches $D_t$, and then $C_\rma$ becomes small. However, if motion were expensive ($\alpha \gg 1$), then the optimal solution would tolerate a higher assignment cost while redistributing resources slowly. The length of the time horizon $T$ will also factor into the trade-off between these two costs. 
Note that while the expression~\req{J_def} does not explicitly involve $D_t$, it does depend on it through the marginalization constraint. To clarify, 
we state the problem formally. 
\begin{problem}\label{gen_model}
	Given an initial resource distribution $R_0$ and a time-varying demand distribution $D$, solve
	\begin{align} 
		\inf_{R,V,\K} \, \smint{0}{T} \Big( C_\rma \left( \K_t \right)& +\alpha^2  C_\rmm \left( R_t,V_t \right)  \rom  \Big) \, dt ,			 \\
		\mbox{s.t.} \quad											
		\partial_t R_t(x) =& -\nabla \cdot \left( R_t(x) \, V_t(x) \right) ,						\label{opt1_state.eq}			\\
		\marg \K_t =& (R_t,D_t) 		.											\label{opt1_marg.eq}
	\end{align} 
\end{problem}
In this problem, $D$ is assumed to be a given external signal which induces a constraint~\req{opt1_marg} on any assignment plan $\K$. $V$ and $\K$ are decision variables, while $R$ can be thought of as the ``state'' whose dynamics are given by~\req{opt1_state}. Figure \ref{pict_model} shows a pictorial representation of this model.

\begin{figure}[t]
	\centering
	
	\includegraphics[width=.8\linewidth]{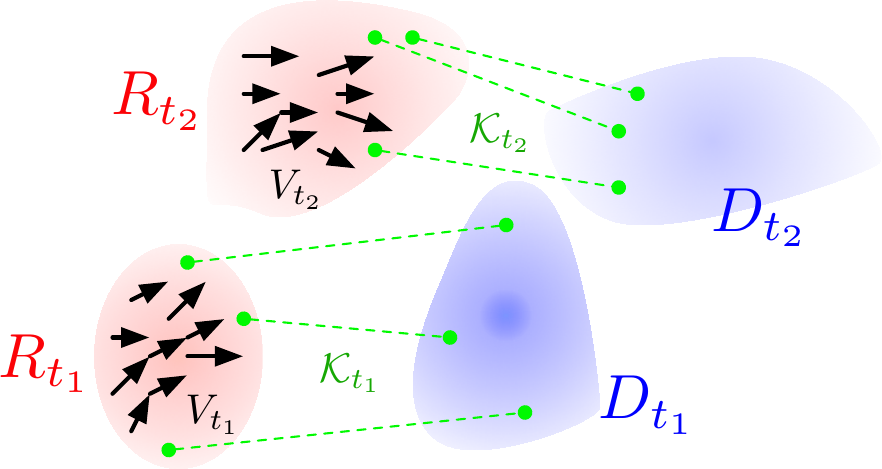}
	
	\mycaption{Depiction of the basic problem formulation. A time-varying demand distribution $D_t$ offloads tasks to a resource distribution $R_t$. The ``task assignment'' (depicted by the green lines) is the optimal instantaneous 
		Kantorovich  plan ${\cal K}_t$ with ``communication cost'' $\W_2^2 ( R_t,D_t )$. $R_t$ is transported by the 
		(control) velocity field $V_t$ to track $D_t$,  
		optimally trading off the assignment cost and motion cost. 
	} 
	\label{pict_model}
\end{figure}

The above problem can be simplified since the marginalization constraint~\req{opt1_marg} is a static-in-time rather than a dynamic constraint. 
Indeed, infimization over $\K$ yields
\begin{equation}
	\begin{split}
		&\inf_{R,V,\K} \,
		\smint{0}{T}   \lp
		C_\rma \left( \K_t \right)  +\alpha^2 C_\rmm \left( R_t,V_t \right)    	\rule{0em}{1em} \rp 	dt		  \\
		& \qquad =~ \inf_{R,V}  \, \smint{0}{T}  \lp  \W_2^2(R_t,D_t) +\alpha^2 C_\rmm \left( R_t,V_t \right)  \rule{0em}{1em}  \rp dt ,
	\end{split}
\end{equation}
since we have chosen $C_\rma$ as~\req{Ca_def}, which has optimal value \eqref{Kb_opt.eq} at each time $t$.
We therefore can restate Problem~\ref{gen_model} as a more standard optimal control problem as follows.

\begin{problem}\label{spec_model}
	Given an initial resource distribution $R_0$ and a time-varying demand distribution $D$, solve
	\begin{equation}
		\begin{aligned} 
			\inf_{R,V} \, \smint{0}{T}  & \lp \W_2^2 \left( R_t,D_t \right)  + \alpha^2 \! \smint{\Omega}{} \left\| V_t(x) \right\|_2^2 R_t(x) \, dx  \rp dt  \\
			\mbox{s.t.} \quad
			&\partial_t R_t(x) = -\nabla \cdot \left( R_t(x) \, V_t(x) \right) .
		\end{aligned}
		\label{opt2.eq}
	\end{equation}
\end{problem}

Since the quantity $\W_2^2(R_t,D_t)$ is a measure of distance between the state $R_t$ and the signal $D_t$, the first term in the objective can be interpreted as a tracking error. The second term can be interpreted as a control energy, which in this context is the physical energy of moving the resource agents. With these interpretations, the optimal control problem~\req{opt2} is an infinite-dimensional, nonlinear, tracking control (or ``servo-mechanism'') problem~\cite{Sage1977}.

The necessary conditions for optimality for this problem can be written as follows:
\begin{align}
	\partial_t R_t &= -\nabla \cdot \left( R_t \, \nabla \lambda_t \right),		&&  \hspace{-3mm} R_0 ~ \text{given} , \label{necc_cond_1} \\
	\partial_t \lambda_t &= - \half \| \nabla \lambda_t \|_2^2 + \tfrac{1}{2 \alpha^2} \tfrac{\delta}{\delta R_t} \W_2^2(R_t,D_t), && \lambda_T = 0 . \label{necc_cond_2}
\end{align}
Here, \eqref{necc_cond_1} is the original dynamics with $V = \nabla \lambda$, while \eqref{necc_cond_2} is the costate equation describing the evolution of the costate (i.e. Lagrange multiplier) $\lambda: \Omega \times [0,T] \to \R$. The term $\tfrac{\delta}{\delta R_t}$ denotes the first variation with respect to $R_t$. We do not include the derivation here because it is not important for our development. Rather, we only wish to emphasize the challenge associated with a direct numerical solution.

The system \eqref{necc_cond_1}-\eqref{necc_cond_2}  forms a nonlinear \emph{two-point boundary value problem}, i.e., the equations for $R$ and $\lambda$ are nonlinearly coupled, with $R$ specified at $t=0$ and $\lambda$ specified at $t=T$. Such systems do not admit solution by direct numerical integration (as an initial value problem would), and typically need to be solved iteratively. Furthermore, in dimensions 2 and higher, the term $\tfrac{\delta}{\delta R_t} \W_2^2(R_t,D_t)$ must itself be found as the solution of an optimization problem. Consider that this term must be computed at each point in time on every iteration, and it is seen why a direct numerical solution is challenging.

Despite the challenge of the general case, in one dimension, the optimal assignment is very structured. This allows not only for simpler numerical methods, but in fact for explicit solutions, which will be the focus of Section \ref{1D}.

\subsection*{Comparison with Dynamic Optimal Transport}

Before moving on, we find it useful to compare Problem \ref{spec_model} posed above with that of \emph{dynamic optimal transport}, originally posed by Benamou and Brenier in \cite{Benamou2000}. This latter problem is written (using our notation) as follows: given an initial state $R_0$ and a \emph{fixed} demand $D$, solve
\begin{equation} \label{dynamic_ot_eq}
	\begin{aligned} 
		\inf_{R,V} \, \smint{0}{1}  & \smint{\Omega}{} \left\| V_t(x) \right\|_2^2 R_t(x) \, dx \, dt  \\
		\text{s.t.} \quad &\partial_t R_t(x) = -\nabla \cdot \left( R_t(x) \, V_t(x) \right) , \\
		& R_T = D .
	\end{aligned}
\end{equation}

Indeed, our Problem 2 bears a close resemblance to \eqref{dynamic_ot_eq}: our motion cost is exactly this objective, and the dynamics for the two problems are identical. Furthermore, the value of \eqref{dynamic_ot_eq} is known to be the squared 2-Wasserstein distance $\W_2^2(R_0,D)$, which tells us that our particular choice for $C_\rma$ and $C_\rmm$ in Problem \ref{spec_model} is very special. That said, the two problems differ in one important respect: while the dynamic optimal transport problem \eqref{dynamic_ot_eq} has a hard constraint on the final state $R_T = D$, our Problem \ref{spec_model} has a penalty on the distance $\W_2^2(R_t,D_t)$ which is integrated over time. This has several consequences. First, it allows Problem \ref{spec_model} to accommodate time-varying demands, while the demand in \eqref{dynamic_ot_eq} must be fixed. Second, it allows Problem \ref{spec_model} to accommodate demands which may not strictly be reachable (e.g. continuous demands when the resource is discrete), while the demand in \eqref{dynamic_ot_eq} must be reachable for solutions to exist. Third, it incentivizes the resource in Problem \ref{spec_model} to be close to the demand at all times, not just at the final time $T$ as in \eqref{dynamic_ot_eq}. This can be summarized by saying that while the dynamic optimal transport problem \eqref{dynamic_ot_eq} is an optimal state-transfer control problem, Problem \ref{spec_model} is an optimal tracking control problem. Thus, while closely related, these two problems have a distinctly different character.

This also has consequences for numerical methods, which can be seen by comparing the necessary conditions for optimality of Problem \ref{spec_model} with those of \eqref{dynamic_ot_eq}, written as follows:
\begin{align}
	\partial_t R_t &= -\nabla \cdot \left( R_t \, \nabla \lambda_t \right),		&& R_0, R_T ~ \text{given} \label{dynamic_necc_cond_1} \\
	\partial_t \lambda_t &= - \half \| \nabla \lambda_t \|_2^2 . \label{dynamic_necc_cond_2}
\end{align}
While again bearing a close resemblance to \eqref{necc_cond_1}-\eqref{necc_cond_2}, there are important differences. First, unlike \eqref{necc_cond_1}-\eqref{necc_cond_2}, the equations \eqref{dynamic_necc_cond_1}-\eqref{dynamic_necc_cond_2} are only coupled one way: $\lambda$ affects the evolution of $R$, but not vice versa. Second, there is no term $\tfrac{\delta}{\delta R_t} \W_2^2$ in \eqref{dynamic_necc_cond_1}-\eqref{dynamic_necc_cond_2} requiring the solution of an optimization problem. This allows \eqref{dynamic_necc_cond_1}-\eqref{dynamic_necc_cond_2} to be solved in a relatively straightforward manner compared to \eqref{necc_cond_1}-\eqref{necc_cond_2}.


\section{Solution to the One-Dimensional Case} \label{1D}

In this section, we solve our proposed model in the special case where the spatial domain is one-dimensional. This is done in two main steps. First, in Section \ref{transformation_section}, we show how the problem~\req{opt2} can be transformed into an equivalent infinite-dimensional LQ tracking problem. Second, in Section \ref{decoupling_section}, we show how this infinite-dimensional LQ tracking problem can be decomposed into an infinite number of uncoupled scalar LQ tracking problems. The scalar LQ tracking problem has a well-established solution, which is reviewed in Section \ref{lqprob}. In Section \ref{solution_section}, we transform these solutions back to our original problem setting to obtain optimal controls for the original problem.

\subsection{Transformation using Quantile Functions} \label{transformation_section}

In this section, we show how the problem \eqref{opt2.eq} can be transformed into an equivalent infinite-dimensional LQ tracking problem when formulated in terms of the \emph{quantile functions}\footnote{This transformation to quantile functions is well-known in the optimal transport literature for the 1D optimal transport problem, and arises naturally in the construction of the monotone transport plan (see, e.g., \cite[Chapter 2]{Santambrogio2015}). It happens that this transformation -- which ``decouples'' the classical optimal transport problem -- decouples our problem as well.} of the densities $R$ and $D$.
\begin{defn}
	Let $\mu$ be a density in $\mathfrak{D}(\Omega)$ with $\Omega \subset \R$. The {\em cumulative distribution function (CDF)} $F_\mu:\Omega\rightarrow[0,1]$ and {\em quantile function} $Q_\mu:[0,1]\rightarrow\Omega$ of $\mu$ are defined by
	\begin{align}
		F_\mu(x) &~:=~ \smint{-\infty}{x} \mu(\xi)\,d\xi	,				\label{cdf_def}				\\
		Q_\mu(z) &~:=~  \inf \{ x : F_\mu(x) \geq z \}. 							\label{quantile_def}
	\end{align}
\end{defn}
We recall the following facts about $F_\mu$ and $Q_\mu$ \cite{Kaempke2015}:

\begin{enumerate}
	\item The associations between $\mu, ~ F_\mu, ~ \text{and} ~ Q_\mu$ are 1-1,
	\item $F_\mu$ and $Q_\mu$ are pseudoinverses on $\Omega$ and $[0,1]$,
	\item $F_\mu$ and $Q_\mu$ are both monotone nondecreasing,
	\item $F_\mu$ can be recovered from $Q_\mu$ by the relation
	\begin{equation}
		F_\mu(x) ~=~ \sup \, \{ z : Q_\mu(z) \leq x \} .
		\label{CDF_from_Q.eq}
	\end{equation}
\end{enumerate}

In one dimension, the optimal assignment is monotone, which allows one to express the 2-Wasserstein distance in closed form in terms of quantile functions as follows.

\begin{lem}[{\cite[Proposition~2.17]{Santambrogio2015}}] \label{wass_lem} 
	Let  $\mu$ and $\nu$ be any two densities in $\D(\Omega)$ with $\Omega \subset \R$. The 2-Wasserstein distance between $\mu$ and $\nu$ is given by
	\begin{equation} \label{wass.eq}
		\W_2^2(\mu,\nu) ~=~ \smint{0}{1} (Q_{\nu}(z) - Q_{\mu}(z))^2\,dz
	\end{equation}
	where  $Q_{\mu},\,Q_{\nu}: [0,1] \rightarrow \Omega$ are the quantile functions of the densities $\mu$ and $\nu$ respectively.
\end{lem}

In other words, the bijection $\mu \leftrightarrow Q_\mu$ is an isometry with respect to the metrics $\W_2$, $L^2$. This will become important.

The simplicity of the above form for $\W_2$ compared to the linear program formulation~\req{Kb_opt} suggests that this transformation to quantile functions may be useful for our problem as well, and motivates the reformulation of Problem \ref{spec_model} entirely in terms of quantile functions. We first address the transformation of the dynamics and then the transformation of the objective.

\subsubsection{Transformation of the Dynamics}

Here, we show\footnote{Note that results similar to Lemmas \ref{equivalent_dynamics} and \ref{equivalent_constraints_lem} exist throughout the optimal transport literature. We carry out the derivation in our slightly altered setting in order to keep the paper self-contained -- we do not claim their originality.} that the ``bilinear dynamics'' of $R$~\req{advec} can be transformed into {\em linear additive} dynamics of the quantile function $\QR$ plus some additional constraints. The first step is as follows.

\begin{lem} \label{equivalent_dynamics}
	A time-varying density $R$ evolves according to the one-dimensional continuity equation
	\be
	\partial_t  R(x,t) ~=~ -\partial_x \big( V(x,t) ~R(x,t) 	\rom \big) 									\label{R_PF.eq} 
	\ee
	if and only if its CDF $\FR$ and its quantile function $\QR$ evolve according to the equivalent dynamics
	\begin{align} 
		\partial_t \FR(x,t) &~=~ - V(x,t) ~\partial_x \FR(x,t)											\label{F_PF.eq}			\\ 
		\partial_t \QR(z,t) &~=~ V \!\lp \QR(z,t) , t \romn \rp		.									\label{Q_PF.eq}
	\end{align} 	
\end{lem}

\begin{proof} 
	See Appendix \ref{equivalent_dynamics_appen}.
\end{proof}

The above proposition can be understood as follows. Mass densities $R$ evolving under a velocity field $V$ obey the continuity equation~\req{R_PF}. By integration, their respective CDFs $\FR$ evolve according to the advection equation~\req{F_PF}. Using the fact that $\QR$ and $\FR$ are inverses, we differentiate the relation
\begin{equation}
	\QR(\FR(x,t),t) ~=~ x
\end{equation}
with respect to time and space to deduce the dynamics \eqref{Q_PF.eq} for $\QR$. This argument works when all quantities are differentiable. The details in the general case are given in Appendix \ref{equivalent_dynamics_appen}.

Next, we reparameterize the velocity field by defining
\be
U(z,t) := V\lB \QR(z,t),t  \romn\rB 
\hstm \Leftrightarrow \hstm 
U := V \circ \QR.
\ee
In the language of differential geometry, $U$ is the \emph{pullback} of $V$ to the domain $[0,1]$ by the quantile function $\QR$.
The dynamics~\req{Q_PF} of quantiles then take on the simple form
\begin{equation} \label{linear_dynamics}
	\partial_t \QR(z,t) ~=~ U(z,t),
\end{equation}
provided that we are able to recover $V$ from $U$ by 
\be
V ~=~ \lB V\circ \QR \rB \circ Q_{\scriptscriptstyle R}^{\sm1}  ~=~ U \circ \FR. 
\ee
This imposes constraints on $U$ which we now state precisely.

\begin{figure}[t]
	\centering
	\includegraphics[width=0.9\linewidth]{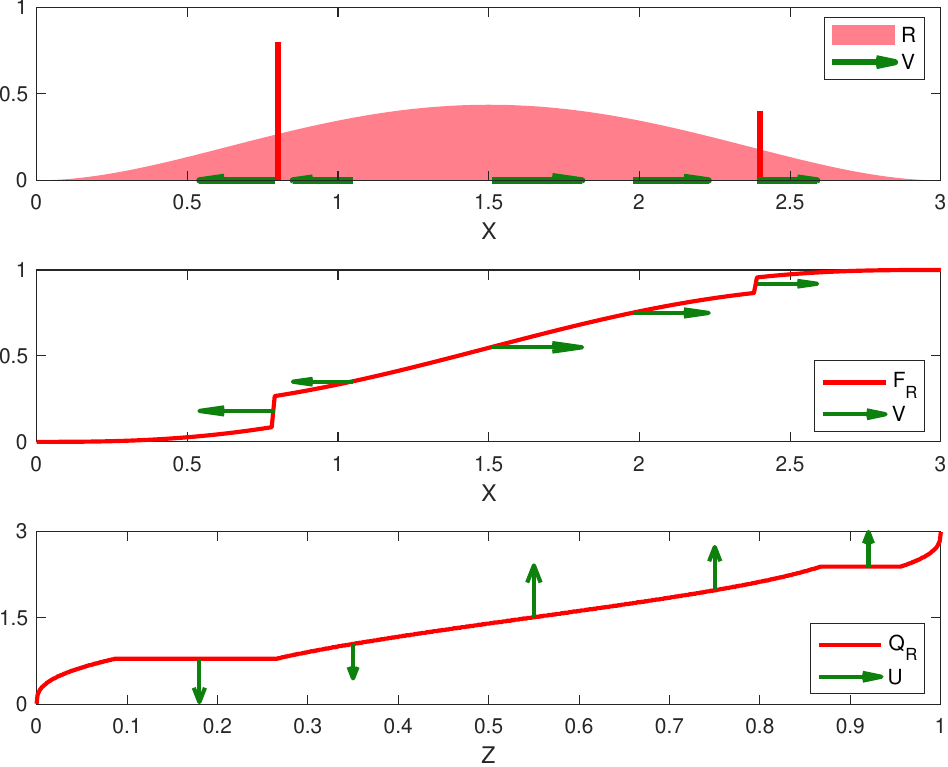}
	\mycaption{(\emph{Top}) A density $R$ is transported by a velocity field $V$ according to the continuity equation~\req{R_PF}. Equivalently, its CDF $\FR$ (\emph{middle}) is advected by $V$ according to the advection equation~\req{F_PF}. 
		(\emph{Bottom}) The corresponding quantile function $\QR(z,t)$ evolves independently at each $z$ 
		with a derivative of $U(z,t) := V \!\lp \QR(z,t), t \romn\rp$.
		Notice that Dirac masses in $R$ correspond to regions where $\QR$ is constant. Since Dirac masses must move with a single velocity, constant regions in $\QR$ must move with a single velocity as well.
	}
	\label{trans_graphic}
\end{figure}

\begin{lem} \label{equivalent_constraints_lem}
	The following two sets are in 1-1 correspondence:
	\begin{equation*}
		\begin{array}{rcl}
			(R,V) & \text{s.t.} &
			\begin{cases}
				\partial_t R = - \nabla \cdot (V R)
			\end{cases} ,
			\vspace{2mm} \\
			(\QR,U) & \text{s.t.} &
			\begin{cases}
				\partial_t \QR = U , \\
				z_1 \leq z_2 ~ \Rightarrow ~ \QR(z_1,0) \leq \QR(z_2,0), \\
				\QR(z_1,t) = \QR(z_2,t) \\
				\qquad \Rightarrow ~ U(z_1,t) = U(z_2,t)
			\end{cases} .
		\end{array}
	\end{equation*}
	The sets are related by the bijective transformation
	\be
	\begin{aligned}
		\mathcal{T}&: (R,V) \mapsto (\QR,V \of \QR) = (\QR,U) \\
		\mathcal{T}^{\sm1}&: (\QR,U) \mapsto (R,U \of \FR) = (R,V) .
	\end{aligned}
	\label{transformation_T.eq}
	\ee
\end{lem}

\begin{proof}
	See Appendix \ref{equivalent_constraints_appen}.
\end{proof}

The constraints on $(\QR,U)$ can be interpreted as follows. The first constraint -- dubbed the \emph{dynamic constraint} -- is just the equivalent dynamics. The second constraint -- dubbed the \emph{monotonicity constraint} -- ensures that the initial state $\QR(\cdot,0)$ is indeed an actual quantile function with a corresponding density $R$. The third constraint -- dubbed the \emph{input constraint} -- ensures that particles that are at the same spatial location move with the same velocity. Recall that this constraint was implicit in the original formulation of the dynamics \eqref{advec.eq}. Together, the second and third constraints imply that $\QR$ remains an actual quantile function for all time, ensuring that the transformations $\mathcal{T}$ and $\mathcal{T}^{-1}$ remain well-defined. The precise details can be found in Appendix \ref{equivalent_constraints_appen}.
Figure~\ref{trans_graphic} illustrates Lemmas~\ref{equivalent_dynamics} and~\ref{equivalent_constraints_lem}.

\subsubsection{Transformation of the Objective}

To complete the statement of the transformed problem, then, we need to rewrite the objective function in terms of $\QR$ and $U$. Lemma \ref{wass_lem} gives us an expression for the assignment cost, but we still need an expression for the motion cost. This is found using a fact about $\QR$ involving the {\em measure pushforward}.

\begin{defn}[Pushforward]
	Let $\mu$ be a density on $\Omega_1$ and $f:\Omega_1\rightarrow\Omega_2$ be a measurable function. 
	The {\em pushforward of $\mu$ through $f$} is a density $f_\# \mu$ on $\Omega_2$ with the property
	\begin{equation} \label{pushfw_eqn}
		\smint{\Omega_2}{} \psi(y) \, \big( f_\# \mu \big) (y) \, dy ~=~ \smint{\Omega_1}{} \big( \psi \of f \big) (x) \, \mu(x) \, dx
	\end{equation}
	for all measurable functions $\psi: \Omega_2 \to \R$.
\end{defn}

The pushforward can be conceptualized as the density formed by ``moving the mass in $\mu$ according to $f$''. That the pushforward is well-defined can be found in any standard text on measure theory, e.g. \cite[Section 3.6]{Bogachev2007}. The pushforward is used in the following fact about quantile functions.

\begin{lem}[{\cite[Proposition~2.2]{Santambrogio2015}}] \label{Q_pushfw_lem}
	Let $\mu$ be a density over $\Omega \subset \R$ and $Q_\mu$ be its quantile function. Then $\mu ~=~ [Q_\mu]_{\sss\!\#} \bone$, where $\bone$ is the uniform density over the unit interval $[0,1]$.
\end{lem}

Using this result, we can rewrite the motion cost as follows
\begin{align}
	&~\smint{\Omega}{} \big\|V(x,t) \big\|^2 \, R_t(x) \, dx 
	= \smint{\Omega}{} \big\| V(x,t) \big\|^2 \, \big( \QR(\cdot,t)_{\sss\!\#} \bone \big) (x) \, dx \nonumber \\
	&=\!\smint{0}{1} \big\| V(\cdot,t) \of \QR(z,t) \big\|^2 \, \bone (z) \, dz = \smint{0}{1} \big\| U(z,t) \big\|^2 \, dz . 
	\label{motion_cost_eq}
\end{align}
Remarkably, the state-dependent weighting term has disappeared. Putting all these pieces together, we can now write our original problem in an equivalent form as follows.

\begin{prop} \label{trans_mod_equiv}
	In one spatial dimension, Problem \ref{spec_model} is equivalent to the following: given an initial resource distribution $R_0$ and a time-varying demand distribution $D$, solve
	\begin{align}
		&\hspace{5em}
		\mathclap{\inf_{\QR,U} \! \smint{0}{T} \!\!\smint{0}{1}  
			\Big( \big(  \QR(z,t) -\QD(z,t) \big)^2\! \text{+}\, \alpha^2 U^2(z,t) \Big) dz\,dt 		}
		\label{trans_mod_one.eq}		 \\
		\text{s.t.} ~~ \partial_t \QR(z,t) &= U(z,t)	,				\label{trans_mod_two.eq}		\\
		\QR(z_1,t) &= \QR(z_2,t) ~ \Rightarrow ~ U(z_1,t) = U(z_2,t) ,		\label{trans_mod_three.eq}	
	\end{align}
	with initial condition $\QR(\cdot,0)$. The solutions to the two problems are related by the transformations~\req{transformation_T}, and the costs attained for each problem are identical.
\end{prop}

\begin{proof}
	It suffices to show that the bijective transformation \eqref{transformation_T.eq} preserves the cost of solutions. This is an immediate application of Lemma \ref{wass_lem} and Equation \eqref{motion_cost_eq}. Note that the construction of $\QR(\cdot,0)$ from $R_0$ automatically satisfies the monotonicity constraint from Lemma \ref{equivalent_constraints_lem}.
\end{proof}

\subsection{Decoupling into Scalar Problems} \label{decoupling_section}

In this section, we show how the infinite-dimensional LQ tracking problem of Proposition \ref{trans_mod_equiv} can be decomposed into an infinite number of uncoupled scalar LQ tracking problems. First, observe that except for the input constraint~\req{trans_mod_three}, the dynamics are decoupled in the index $z$. In addition, the objective~\req{trans_mod_one} is an integration over all $z$. Thus it appears that the solution is to minimize the objective independently at each value of $z$. This is effectively what we will do after accounting for the constraint~\req{trans_mod_three}. Note that at each $z$, the problem is a standard, scalar, LQ tracking problem. 

We account for the constraint~\req{trans_mod_three} as follows. First, observe that the constraints \eqref{trans_mod_two.eq},~\req{trans_mod_three} imply that if the initial quantile $\QR(\cdot,0)$ is equal at two points $z_1$ and $z_2$, then $\QR(z_1,t)=\QR(z_2,t)$  for all $t\geq0$. In other words, the level sets of $\QR$ do not change with time. Furthermore, the values of $\QR(\cdot,t)$ and $U(\cdot,t)$ must be constant over each level set. Thus, by partitioning the domain $[0,1]$ into level sets of $\QR(\cdot,0)$, we can write a single scalar LQ tracking problem over each level set. Then, these problems can be solved independently, as the constraint \eqref{trans_mod_three.eq} is automatically satisfied by the partition\footnote{The careful reader will notice that for this argument to hold, we also cannot form new level sets. While admissible solutions do not satisfy this in general, \emph{optimal} solutions do, which is all that we need for equivalence. This point is treated more carefully in the proof of Proposition \ref{decoup_mod_equiv}.\label{inactive_optimum}}.

We first describe the partitioning of the domain into level sets, as this forms the basis for our solution technique. We give a precise characterization of the problem equivalence later.

\begin{defn}\label{partition_def}
	Given a monotone function $q: [0,1] \to \R$, a {\em $q$-level-set partition $\caP_\rmq$} of $[0,1]$
	is a set of disjoint subsets 
	\be
	\Pq ~:=~ \lcb P_i\subseteq [0,1]; ~i\in \caI \rcb, 
	\ee
	indexed by the set of values $\caI := \text{range}(q)$ such that
	\begin{equation}
		z \in P_i \qquad \Leftrightarrow \qquad q(z) = i .
	\end{equation}
\end{defn}

Note that the index $i$ ranges over all distinct values of the function $q$, which may be finite, or countably or uncountably infinite. 
Elements of $\Pq$ are disjoint since $q:[0,1]\rightarrow\R$ is a single-valued function, and $[0,1]$ is  the union of all 
elements of $\Pq$ since each $z\in [0,1]$ must belong to the set $P_{q(z)}$.

\begin{defn} 							\label{PWC_P.def}
	Let $\Pq := \lcb P_i\subseteq[0,1]; ~i\in\caI \rcb$ be a $q$-level-set partition of $[0,1]$. 
	A function $g:[0,1]\rightarrow\R$ is called {\em $\Pq$-piecewise-constant} if it is constant on each $P_i\in\Pq$, 
	and we write $g(P_i)$ unambiguously for such functions.
	
	For a function $g$ that is not $\Pq$-piecewise-constant, 
	we define its {\em average with respect to the partition  (w.r.t.p.)   $\Pq$} by 
	\be
	\Pav{g} (z) := 
	\left\{ \begin{array}{ll} 
		g(z),  &\mbox{if}~ z\in P_i, ~ \left| P_i \right| = 0 , 		\\ 
		\tfrac{1}{|P_i|} \smint{P_i}{} g(z) ~dz,  & \mbox{if} ~z\in P_i ,  ~    \left| P_i \right| > 0,
	\end{array} 	\right. 
	\label{wrtp_P.eq}
	\ee
	where $\left| P_i \right|$ is the Lebesgue measure of the set $P_i\subseteq [0,1]$.
	Note that $\Pav{g}$ is $\Pq$-piecewise-constant and so $\Pav{g}(P_i)$ is unambiguous.
\end{defn}

We will mainly average quantile functions over partitions of $[0,1]$. Given a quantile function $Q_\mu$ of a 
density $\mu$, its average $\Pav{Q}_\mu$ is the quantile of some  density (denoted by $\Pav{\mu}$)
which can be reconstructed from $\Pav{Q}_\mu$ using Lemma \ref{Q_pushfw_lem}.
Notice that since $Q_\mu$ is monotone, each $P_i$ is either a singleton (if $|P_i| = 0$) or a finite interval (if $|P_i| > 0$).
If $P_i$ is a singleton, then $\Pav{\mu}(i) = \mu(i)$. If $P_i$ is a finite interval, then $\Pav{\mu}(i)$ is a single Dirac mass corresponding to the mass in $\mu$ on $Q_\mu(P_i)$.
A depiction of all of these ideas is illustrated in Figure~\ref{partition_graphic}.

\begin{figure}[t]
	\centering
	\includegraphics[width=0.9\linewidth]{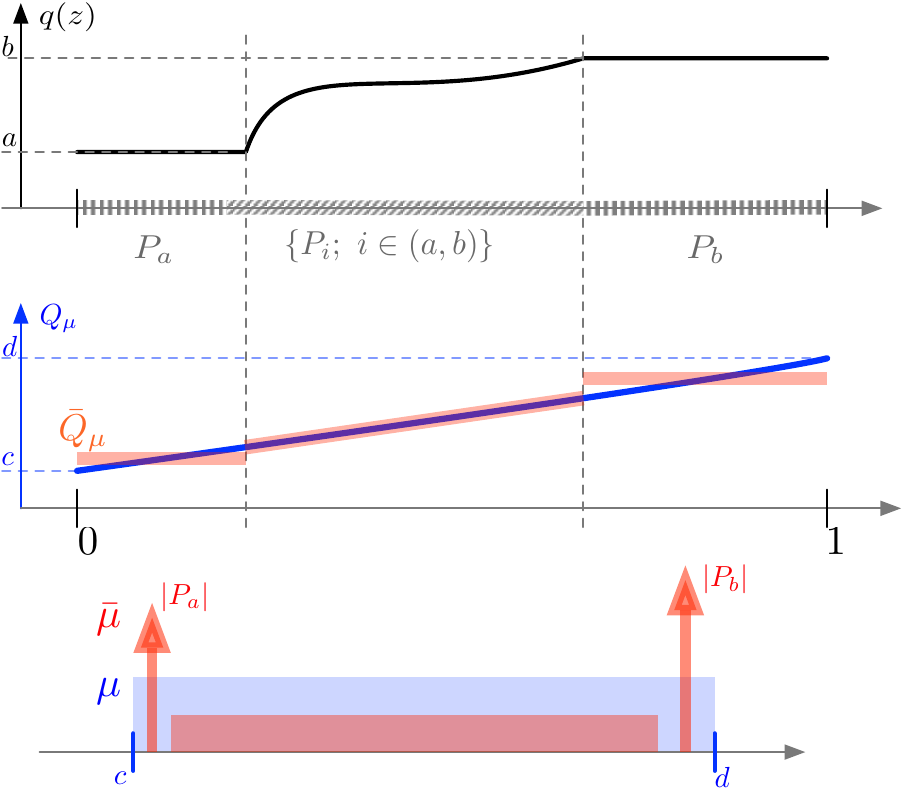}
	
	\mycaption{({\em Top}) The $q$-level-set partition $\Pq$ (Definition~\ref{partition_def})
		of the interval $[0,1]$ into level sets of a function $q:[0,1]\rightarrow \R$ (black).
		The partition is composed of the sets $P_a$, $P_b$ (finite intervals, 
		corresponding to the levels $a$ and $b$ of $q$), and 
		a continuum of singleton level sets $\lcb P_i; ~i\in(a,b) \rcb$
		of the values of $q$ between $a$ and $b$. 
		({\em Bottom}) A density $\mu$,  its ({\em middle}) quantile function $Q_\mu$, and 
		its average $\Pfav{Q}_\mu$  with respect to the partition $\Pq$. The average is 
		$\caP$-piecewise-constant (Definition~\ref{PWC_P.def}). ({\em Bottom}) $\Pav{\mu}$ is the 
		density of the averaged quantile  $\Pfav{Q}_\mu$. The Dirac masses, each of mass
		$\left| P_i \right|$, correspond to each $P_i$ of non-zero Lebesgue measure in the partition $\caP$.
	}
	\label{partition_graphic}
\end{figure}

We will now use these tools to define the equivalent family of scalar LQ tracking problems.
Let $\caP:=\caP_{\QR(\cdot,0)}$ be the $\QR(\cdot,0)$-level-set partition of  the interval  $[0,1]$, and denote the index set of this partition by $\caI := \text{range}(\QR(\cdot,0))$.
Recall that the constraints \eqref{trans_mod_two.eq},~\req{trans_mod_three} imply that $\QR(\cdot,t)$ and $U(\cdot,t)$ remain constant over each $P_i \in \caP$ for all $t$. In other words, $\QR(\cdot,t)$ and $U(\cdot,t)$ are $\caP$-piecewise-constant. The values $r_i(t) := \QR(P_i,t)$ and $u_i(t) := U(P_i,t)$ will thus be taken to be the state and input for the $i^{th}$ scalar LQ tracking problem.

The demand signal $\QD(\cdot,t)$ is not $\caP$-piecewise-constant. However, we claim that we can replace it with $\QDb(\cdot,t)$, its average w.r.t. the partition $\caP$, without changing the optimal solutions. This claim is justified by the following observations. First, if $f$ is a constant function on a set $S$ and $g$ is an integrable function on $S$, then
\begin{equation}
	\smint{S}{} (f(s) - g(s))^2 \, ds ~=~ \smint{S}{} (f(s) - \bar{g}(s))^2 \, ds ~+~ k ,
\end{equation}
where $\bar{g}$ denotes the average value of $g$ on $S$ and $k := \int (g - \bar{g})^2$ is a constant which depends only on $g$. 
Using this fact, we can rewrite the objective \eqref{trans_mod_one.eq} as follows
\begin{align}
	&\smint{0}{T} \smint{0}{1}  \big(  \QR(z,t) -\QD(z,t) \big)^2 + \alpha^2 U^2(z,t) \, dz\,dt 	\\
	&= \smint{0}{T} \smint{0}{1}  \big(  \QR(z,t) -\QDb(z,t) \big)^2 + \alpha^2 U^2(z,t) \, dz\,dt ~+~ K , \nonumber
\end{align}
where $K$ is a constant depending only on $\QD$ and the partition $\caP$ (i.e. not on $\QR$ or $U$).
Since adding a constant to the objective does not change optimal solutions, we are justified in considering the problem with $\QDb$ in place of $\QD$. Furthermore, $\QDb$ is $\caP$-piecewise constant, so we can take $d_i(t) := \QDb(P_i,t)$ to be the tracking signal for the $i^{th}$ scalar problem. 
Now, we can restate the optimal control problem of Proposition~\ref{trans_mod_equiv} as a decoupled family of scalar LQ tracking problems as follows.
\begin{prop} 													\label{decoup_mod_equiv}
	Given the problem of Propsition~\ref{trans_mod_equiv}, 
	let $\caP$ be the $\QR(\cdot,0)$-level-set partition of $[0,1]$ with index set $\caI$,
	and consider the family of scalar LQ tracking problems
	\begin{equation}
		\begin{aligned}
			& \inf_{r_i,u_i} \, \smint{0}{T} \lp  \big( r_i(t) - d_i(t)  \big)^2+\alpha^2 u_i^2(t) \rp dt \\
			& \quad \text{s.t.} \quad \dot{r}_i(t) = u_i(t) ,
			\hstm 
			r_i(0) = \QR(P_i,0). 
		\end{aligned}
		\label{opt4.eq}
	\end{equation}
	for $i \in \caI$, where the $d_i$ are the demand reference signals
	\be
	d_i (t) := \QDb(\cdot,t) (P_i) .
	\label{di_averaged.eq}
	\ee
	The solution to the problem of Proposition~\ref{trans_mod_equiv} is given from the solutions of~\req{opt4} by 
	\begin{align} \label{reconstruction}
		\QR({P}_i,t) = {r}_i(t), 
		\hstm 
		{U}({P}_i,t) = {u}_i(t),
	\end{align}
	and the costs are related by
	\begin{equation} \label{prop_12_cost}
		J(\QR,U;\alpha,T) ~=~ \smint{0}{1} J_i(r_i,u_i;\alpha,T) \, dz ~+~ K,
	\end{equation}
	where $J$ denotes the value of \eqref{trans_mod_one.eq}, $J_i$ the value of the $i^{th}$ scalar problem \eqref{opt4.eq} with $i = \QR(z,0)$, and
	\begin{equation}
		K ~:=~ \smint{0}{T} \smint{0}{1}  \left( \QDb(z,t) - \QD(z,t) \right)^2 \, dz \, dt.
	\end{equation}
\end{prop}

\begin{proof}
	The proof follows directly as outlined above. The equivalence of the dynamics is trivial. The only claim that requires further justification is that the constraint \eqref{trans_mod_three.eq} is inactive at the optimum\footref{inactive_optimum}. The details are in Appendix \ref{decoup_mod_equiv_appen}.
\end{proof}

The above family of scalar LQ tracking problems are indexed by $\caI$, the index of the partition sets as determined by the initial 
resource quantile $\QR(\cdot,0)$. In particular, there is a distinction between the regions where $\QR(\cdot,0)$ is constant versus where it is strictly increasing. In the latter case, the level sets $P_i$ are singletons, and the problems \req{opt4} form a continuum of scalar LQ tracking problems, one for each value of $z$. On the other hand, when the resource distribution is composed of discrete agents, $\QR(\cdot,0)$ is piecewise constant, the level sets $P_i$ are finite intervals, and the problems \req{opt4} form a finite collection of scalar tracking problems, one for each agent of non-zero mass. This  case is depicted graphically in Figure \ref{decoupling}.

\begin{figure}[t]
	\centering
	\includegraphics[width=0.9\linewidth]{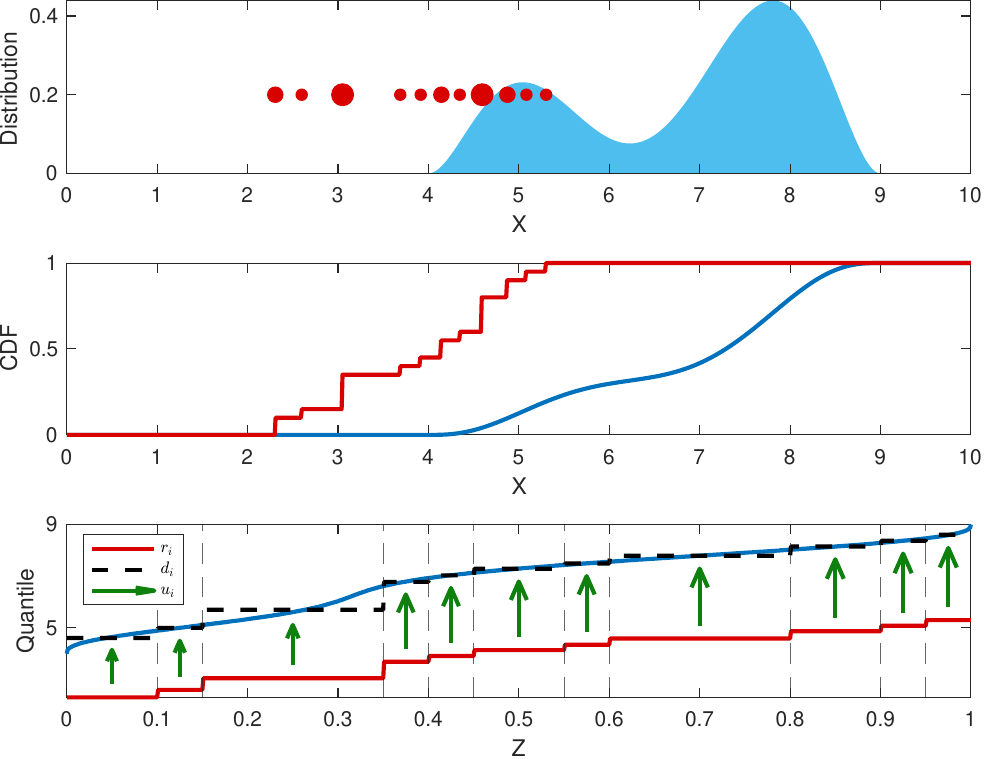}
	\mycaption{Decoupling in discrete-agent case. Each agent in the resource distribution (red) corresponds to a constant region in the quantile function. The vertical dotted lines show the separation of these constant regions via partitioning. A single scalar LQ tracking problem can be written for each element in this partition, with state $r_i = \QR(P_i)$ (red), control $u_i = U(P_i)$ (green), and tracking signal $d_i = \QDb(P_i)$ (black).}
	\label{decoupling}
\end{figure}

\subsection{Solution of the Scalar LQ Tracking Problem} \label{lqprob}

The solution of the scalar LQ tracking problem~\req{opt4} is well-established in the literature (see, e.g. \cite[Section~5.2]{Sage1977}, where it is referred to as the ``servomechanism problem''). We summarize the solution here for our particular setting. 

The optimal control $u_i$ consists of feedback and  feedforward components
\begin{equation}
	u_i(t) ~=~ - \tfrac{1}{\alpha^2} \big(  \ric (t) \, r_i(t)  ~+~ \ff_i(t)  \rom \big), 
	\label{u_opt_form.eq}
\end{equation}
where $\ric$ solves a (scalar) differential Riccati equation and $\ff_i$ is the output of a linear time-varying 
system driven by the ``reference signal'' $d_i$
\begin{align}
	\dot{\ric}(t) ~&=~ \tfrac{1}{\alpha^2}  \ric^2(t)  - 1	,		&& \ric(\sT)=0	,			\label{diff_ricc.eq}	\\
	\dot{\ff}_i(t) ~&=~ \tfrac{1}{\alpha^2} \ric(t) ~ \ff_i(t)  ~+~ d_i(t),		&& \ff_i(\sT)=0 , 			\label{diff_g.eq}	\\
	\dot{r}_i(t) ~&=~ -\tfrac{1}{\alpha^2} \ric(t) ~ r_i(t)  ~-~ \tfrac{1}{\alpha^2}  \, \ff_i(t),		&& r_i(0) ~\mbox{given} 	.
	\label{diff_r_i.eq}
\end{align}
The Riccati equation solution is always positive, and therefore the equation for the feedforward term $\ff_i$ is stable 
evolving backwards, whereas the closed-loop system equation for $r_i$ is stable evolving forwards. 
Note that since   the feedforward equation~\req{diff_g} has a final boundary condition, it
requires  the demand signal $d_i$ to be known ahead of time. 
The scalar  differential Riccati equation~\req{diff_ricc} can be solved explicitly
\be
\ric(t) ~=~ {\alpha}~ \tanh \left( (\sT-t) / {\alpha} \right).
\label{f_explicit.eq}
\ee
This gives the time-varying ``A-matrices'' of the scalar systems~\req{diff_g} and~\req{diff_r_i}. The state transition ``matrix'' of the closed-loop  system~\req{diff_r_i} for $r_i$ is then computed as 
\be
\phi_\rmr(t,\tau) = \exp\lB\sm  \smint{\tau}{t}  \tfrac{1}{\alpha^2}  \ric (s) \,ds \rB
= \frac{\cosh \lB (\sT-t) / {\alpha} \rB}{\cosh \lB (\sT-\tau) / {\alpha} \rB }, 
\label{phi_r_def.eq}
\ee
where the explicit expression comes from substituting~\req{f_explicit}. The state transition matrix for the feedforward term $\ff_i$ is the reciprocal of that, i.e., $\phi_\rmy(t,\tau)=1/\phi_\rmr(t,\tau)$, 
and therefore the variation-of-constants formula gives the solution to~\req{diff_g} as a backwards integration
\be
\ff_i(t) 
=\!\! \smint{T}{t} \! \phi_\rmy (t,\tau) \,d_i(\tau) \,d \tau  		,	~		
\phi_\rmy (t,\tau) 
= 	 \frac{\cosh \lB (\sT\sm \tau) / {\alpha} \rB}{\cosh \lB (\sT\sm t) / {\alpha} \rB }  .
\label{g_integral.eq}
\ee
The exact form of the feedforward term $\ff_i$ then depends on the signal $d_i$ to be tracked. Using the variation-of-constants formula, it is also possible to write expressions for $r_i$, $u_i$, and the cost $J_i$, but these do not simplify in general, and so we do not include them here.

\subsection{Solution to Original Problem} \label{solution_section}

The solutions of the scalar LQ tracking problems can be put together using Proposition~\ref{decoup_mod_equiv}
to give the solution to the problem of 
Proposition~\ref{trans_mod_equiv} in terms of the quantile functions.

\begin{prop} \label{trans_mod_soln}
	The solution to the transformed problem of Proposition~\ref{trans_mod_equiv} is given by
	\begin{align}
		{U}(z,t) ~&=~ -\tfrac{1}{\alpha^2}  \ric(t) ~ \QR(z,t)  ~-~ \tfrac{1}{\alpha^2}  \sfY(z,t) , 
	\end{align} 
	where the feedback $\ric$ is given by \eqref{f_explicit.eq} and feedforward $\sfY$ by
	\begin{align} 
		\sfY(z,t) &~=~ \smint{T}{t} \phi_\rmy(t,\tau) \, \QDb(\cdot,\tau)(z) \, d \tau 	,
	\end{align}
	where $\phi_\rmy$ is given by~\req{g_integral}, and 
	$\QDb$ is the average of $\QD$ with respect to the $\QR(\cdot,0)$-level-set partition of $[0,1]$.
\end{prop}

This in turn gives the solution to the original Problem~\ref{spec_model}.

\begin{thm} \label{spec_mod_soln}
	In one spatial dimension, the optimal velocity field (control) of Problem~\ref{spec_model} is given by
	\be
	{V}(x,t) ~=~ - \tfrac{1}{\alpha^2} \ric(t) \, x ~-~\tfrac{1}{\alpha^2}  \sfM(x,t) ,
	\ee
	where the feedback $\ric$ is given by~\req{f_explicit} and feedforward $\sfM$ by 
	\be
	\sfM(x,t) ~=~ \smint{T}{t} \phi_\rmy(t,\tau) \,
	\QDb(\cdot,\tau) \big( \FR(x,t) \big) \,d \tau	,	\label{FF_origin.eq}
	\ee
	where $\phi_\rmy$ is given by~\req{g_integral}, and 
	$\QDb$ is the average of $\QD$ with respect to the $\QR(\cdot,0)$-level-set partition of $[0,1]$.
\end{thm}

Note again that the feedforward term $\sfM$ in~\req{FF_origin} requires a backwards integration of the demand
signal $D$ from the final time $\sT$. 
Thus  the entire demand signal must be known ahead of time. There are at least two practical cases where this could reasonably be assumed. The first case is  when the demand is static, and the second is when the demand is periodic in time. The next section investigates each of these special cases.

Before moving on, we point out a curiosity regarding the costs for these problems. As can be seen in \eqref{prop_12_cost}, the minimum costs are bounded below by the constant
\begin{equation} \label{k_eqn}
	K ~=~ \smint{0}{T} \smint{0}{1}  \left( \QDb(z,t) - \QD(z,t) \right)^2 \, dz \, dt.
\end{equation}
This constant depends only on $\QD$ and the partition $\caP$ (which in turn depends on the initial resource quantile $\QR(\cdot,0)$). In other words, $K$ depends only on the problem parameters, and can therefore be interpreted as a \emph{fundamental performance limitation} of the system. I.e., is not possible to achieve a cost lower than $K$ for fixed $D$, $R_0$. However, if we are allowed the choice of $R_0$, then we can affect $\caP$, potentially lowering this cost. This raises an interesting question: given a resource composed of $N$ discrete agents with (unequal) masses and a demand signal $D$, how should the agents be initially positioned so as to minimize $K$? While we do not answer this question here, a heuristic argument based on the graphs of the quantile functions suggests that agents with larger mass be positioned near regions of high demand density and vice versa.


\section{Special Cases: Static and Periodic Demands}					\label{explicit.sec}

In this section, we investigate further the special cases of static and periodic demands.

\subsection{Static Demand} \label{static_case}

When the demand signal $D_t = D$ is static (i.e. constant in time), the feedforward term $\ff_i$ takes on the form
\begin{equation}
	\ff_i(t) ~=~ - \ric(t) \, d_i .
\end{equation}
This gives the optimal control \req{u_opt_form} in ``error feedback'' form 
\begin{equation} \label{input_u}
	u_i(t)  ~=~  -\tfrac{1}{\alpha^2}  \ric(t) \, \big( r_i(t) - d_i \big) .
\end{equation}
The optimal state trajectory can then be computed as 
\begin{equation}
	r_i(t) ~=~ \phi_\rmr(t,0) \, r_i(0) ~+~ \lB 1-\phi_\rmr(t,0) \romn \rB \, d_i .  
	\label{traj.eq}
\end{equation}
This trajectory is a straight-line interpolation between the initial state $r_i(0)$ and the constant reference signal $d_i$, 
where the ``rate of travel'' along this trajectory is  determined  by $\phi_\rmr$. Substituting this trajectory into the cost functional, the optimal cost for each scalar problem is computed as 
\begin{equation}
	J_i (r_i(0),d_i;\alpha,T) ~=~ \lB r_i(0) - d_i \rB^2 \, {\alpha} \, \tanh \big( \sT / {\alpha} \big) . \label{opt_cost_1}
\end{equation}

\begin{figure*}[t]
	\centering
	\begin{subfigure}[t]{0.32\textwidth}
		\centering
		\includegraphics[width=\linewidth]{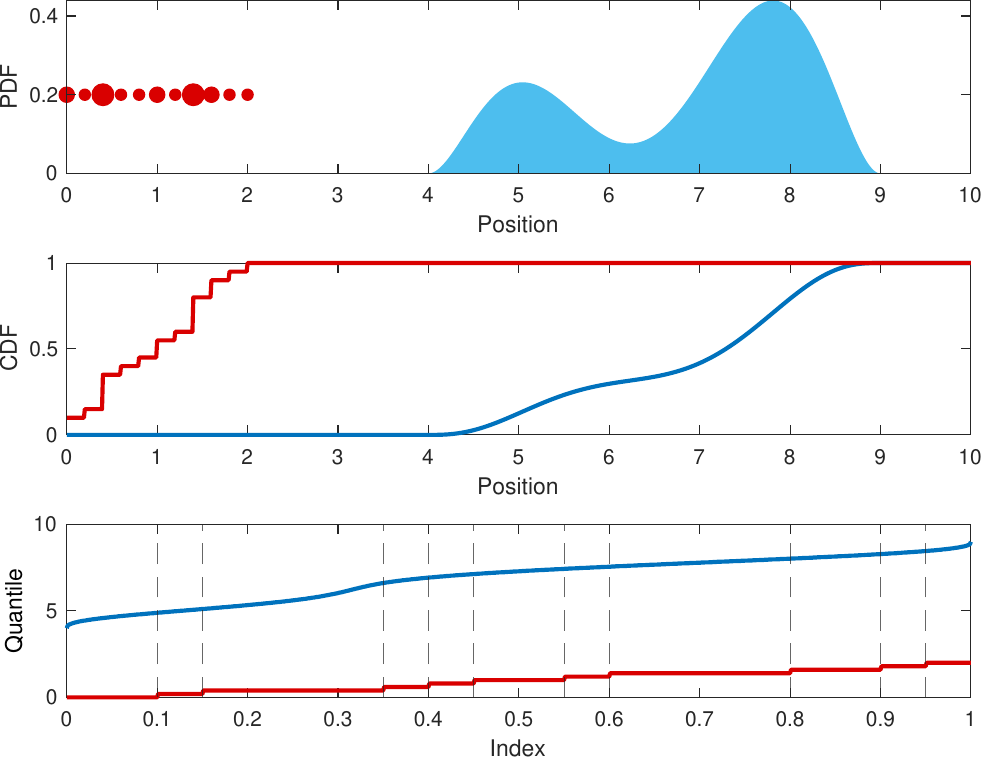}
		\mycaption{Initial conditions ($t=0$)}
		\label{plot_stat_init}
	\end{subfigure}
	\hfill
	\begin{subfigure}[t]{0.32\textwidth}
		\centering
		\includegraphics[width=\linewidth]{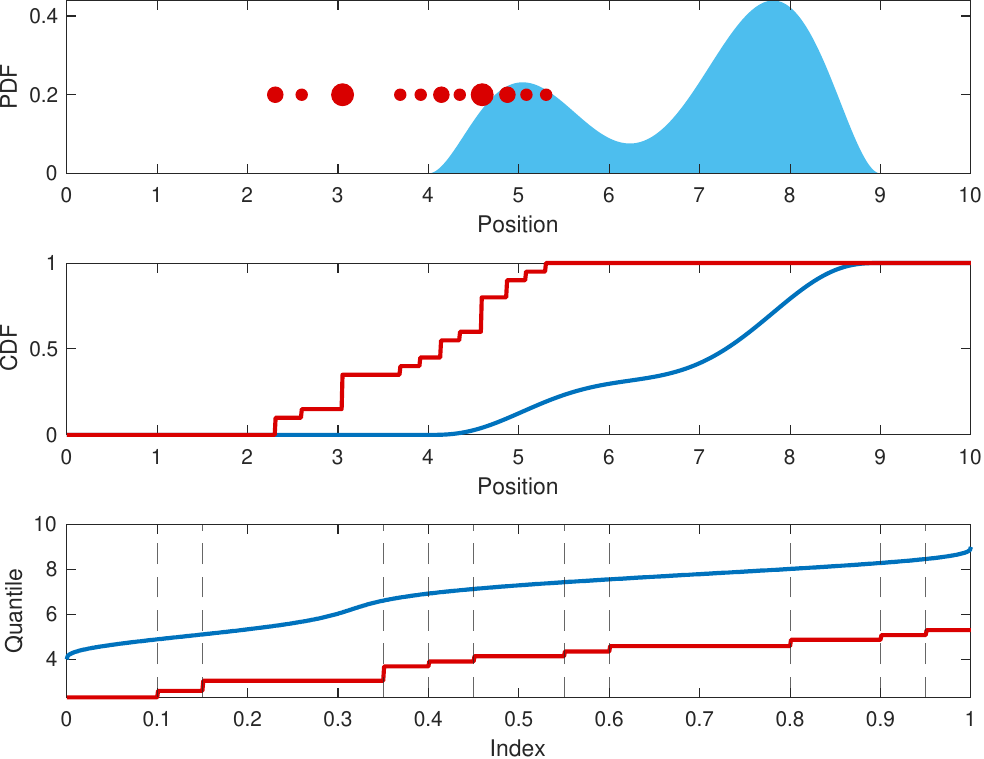}
		\mycaption{Intermediate conditions ($t \approx 1.4$)}
		\label{plot_stat_inter}
	\end{subfigure}
	\hfill
	\begin{subfigure}[t]{0.32\textwidth}
		\centering
		\includegraphics[width=\linewidth]{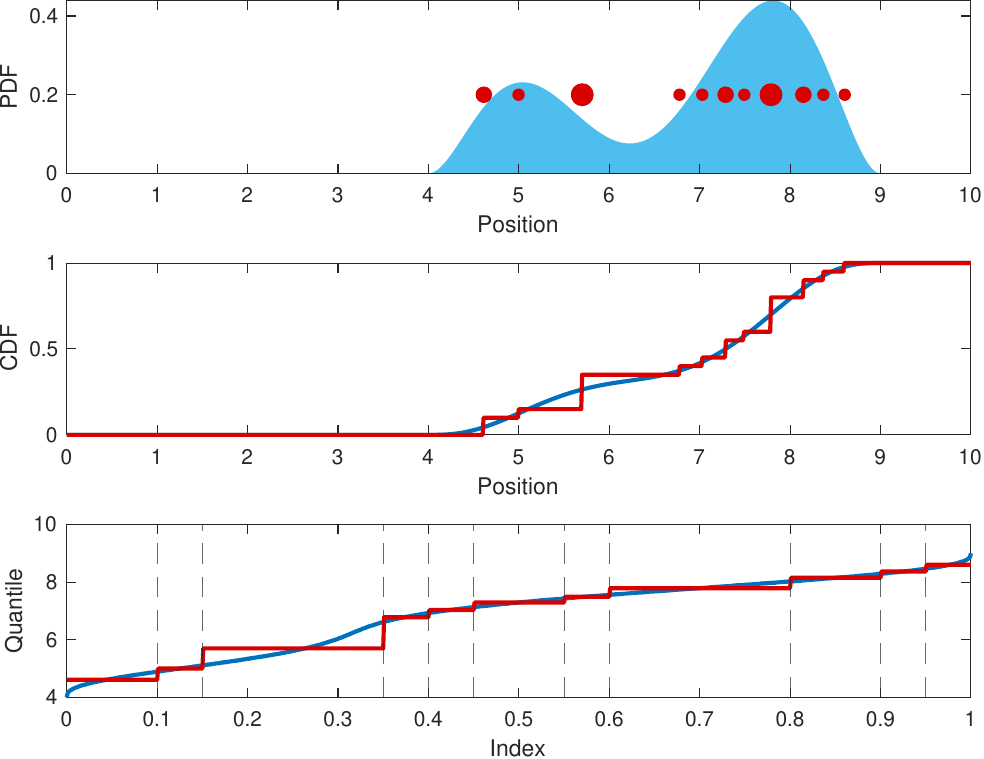}
		\mycaption{Final conditions ($t=10$)}
		\label{plot_stat_final}
	\end{subfigure}
	\mycaption{Initial, intermediate, and final conditions of the resource $R$ (red), demand $D$ (blue), and their corresponding CDFs and quantile functions. The vertical dotted lines on the quantile function plots separate elements of the partition $\mathcal{P}$. The parameters $\alpha = 2$ and $\sT=10$ were used for this particular example. Due to the relatively large value of $\sT / {\alpha}$, the final distribution $R_\sT$ ends up very close to $\bar{D}$, the nearest reachable distribution to $D$.}
	\label{fig:three graphs}
\end{figure*}

Applying Proposition \ref{decoup_mod_equiv} immediately gives the solution to the problem of Proposition \ref{trans_mod_equiv} in this special case.

\begin{prop}
	Consider the problem of Proposition \ref{trans_mod_equiv} with $D$ constant-in-time. The optimal control is given by
	\begin{equation}
		U(\cdot,t) ~=~ -\tfrac{1}{\alpha^2}  \ric(t) \, \big( \QR(\cdot,t) - \QDb \big) ,
	\end{equation}
	where $\ric$ is given by \eqref{f_explicit.eq}, and $\QDb$ is the average of $\QD$ with respect to the $\QR(\cdot,0)$-level-set partition of $[0,1]$.
	This control generates the trajectory
	\begin{equation}
		\QR(\cdot,t) ~=~ \phi_\rmr(t,0) \, \QR(\cdot,0) ~+~ \lB 1-\phi_\rmr(t,0) \romn \rB \, \QDb ,
		\label{traj2.eq}
	\end{equation}
	where $\phi_\rmr$ is given by \eqref{phi_r_def.eq}, and attains the cost
	\begin{equation}
		J =  \big\| \QR(\cdot,0) - \QDb \big\|_{L^2}^2 \, {\alpha} \tanh \big( \sT / \alpha \big) + \, T \big\| \QD - \QDb \big\|_{L^2}^2 .
	\end{equation}
\end{prop}

In particular, the optimal trajectory follows a straight line in the $L^2$ space of quantile functions from $\QR(\cdot,0)$ to $\QDb$ at a rate determined by $\phi_\rmr$. Furthermore, $\QDb$ is the nearest function to $\QD$ in the $L^2$ subspace of $\caP$-piecewise-constant functions. Since the (monotone nondecreasing) $\caP$-piecewise-constant functions are exactly those which are reachable from the initial condition $\QR(\cdot,0)$ under the constraints \eqref{trans_mod_two.eq}, \eqref{trans_mod_three.eq}, we get a very nice form for the overall solution: \emph{optimal trajectories $\QR$ follow the shortest path from the initial state $\QR(\cdot,0)$ to $\QDb$ -- the nearest reachable state to $\QD$ -- at a rate determined by $\phi_\rmr$.} We also emphasize that the optimal cost decomposes into the weighted sum of the distances from $\QR(\cdot,0)$ to $\QDb$ and from $\QDb$ to $\QD$. This latter term is exactly the ``performance limitation'' constant $K$ from the previous section, which can be interpreted here as the distance from $\QD$ to the reachable set. Due to the isometry between the $L^2$ space of quantile functions and the 2-Wasserstein space of densities in 1D, this solution structure carries over exactly to our original problem setting. That is\footnote{Compare this result with that of dynamic optimal transport \cite{Benamou2000}, where solutions follow Wasserstein geodesics from $R_0$ to $D$ at a constant rate.} : \emph{the optimal trajectory $R$ follows the Wasserstein geodesic from the initial state $R_0$ to a density $\bar{D}$ -- the nearest reachable density to $D$ in the $\W_2$-distance -- at a rate determined by $\phi_\rmr$.} This is stated formally in the following theorem.

\begin{thm} \label{static_soln}
	Consider the problem \eqref{opt2.eq} with $\Omega \subset \R$ and $D$ constant-in-time. The optimal control is given by
	\begin{equation} \label{static_control}
		{V}(x,t) 	~=~ -\tfrac{1}{\alpha^2}  \ric(t) \,\big( x - \sfM(x,t) \big)	,	
	\end{equation}
	where $\ric$ is given by \eqref{f_explicit.eq} and
	\begin{equation} \label{static_m}
		\sfM(x,t) 		~:=~ \QDb \big( \FR(x,t) \big)  ,
	\end{equation}
	where $\QDb$ is the average of $\QD$ w.r.t. the $\QR(\cdot,0)$-level-set partition of $[0,1]$.
	This control generates the trajectory
	\begin{equation} \label{orig_trajectory}
		R_t ~=~ \Big[ \phi_\rmr(t,0) \, I ~+~ \big( 1-\phi_\rmr(t,0) \big) \, \sfM(\cdot,0) \Big]_\# R_0 ,
	\end{equation}
	where $\phi_\rmr$ is given by~\req{phi_r_def}, $I$ is the identity map on $\Omega$, and $\#$ denotes the measure pushforward \eqref{pushfw_eqn}.
	Furthermore, $\sfM$ is the optimal transport map taking $R_0$ to $\bar{D}$, and thus \eqref{orig_trajectory} follows the associated Wasserstein geodesic.
	Furthermore, this solution attains the cost
	\begin{equation}
		J ~=~ \W_2^2 \lB R_0,\Pav{D} \romn\rB \, {\alpha} \, \tanh (T / {\alpha}) ~+~ T \, \W_2^2 \lB D,\Pav{D} \romn\rB .	\label{opt_cost_2.eq}
	\end{equation}
\end{thm}

\begin{proof}
	\eqref{static_control}, \eqref{static_m}, and \eqref{opt_cost_2.eq} follow directly from an application of Proposition \ref{trans_mod_equiv} and the observation that the $L^2$ space of quantile functions is isometric to the 2-Wasserstein space of densities. The form of the trajectory \eqref{orig_trajectory} comes from Lemma \ref{Q_pushfw_lem} and properties of the pushforward:
	\begin{align}
		R_t ~&=~ \left[ \QR(\cdot,t) \right]_\# \bone~=~  \left[ \Phi_t \of \QR(\cdot,0) \right]_\# \bone \nonumber \\
		~&=~ \left[ \Phi_t \right]_\#  \big( \left[ \QR(\cdot,0)\right]_\# \bone  \big) ~=~  \left[ \Phi_t \right]_\# R_0 , \\
		\Phi_t ~&=~ \phi_\rmr(t,0) \, I ~+~ \big( 1-\phi_\rmr(t,0) \big) \, \sfM(\cdot,0) .
	\end{align}
	The observations that $\sfM$ is an optimal transport map and $\eqref{orig_trajectory}$ a Wasserstein geodesic come from \cite[Sections~2.1-2.2]{Santambrogio2015} and \cite[Theorem~5.27]{Santambrogio2015}, respectively.
\end{proof}

Figure \ref{fig:three graphs} shows a numerical example with a discrete resource distribution consisting of eleven agents of unequal weight and a continuous static demand distribution.

\subsection{Periodic Demand} \label{periodic_case}

To demonstrate the method on time-varying demands, we now consider the case where the demand is a $\sT$-periodic function, i.e., for all $t\geq 0$, $D_t = D_{t+\sT}$, implying that $d_i(t) = d_i(t+\sT)$ for each $i$. In such problems, one takes an infinite time horizon and the cost as an asymptotic average 
\be
J ~:=~ \lim_{T\rightarrow\infty} \frac{1}{\sT} \smint{0}{T} C \lB R_t, D_t,V_t \rB dt , 
\label{inf_hor_T.eq}
\ee
where the instantaneous cost $C$ is taken as the integrand in~\req{opt2}. The transformations described in Propositions \ref{trans_mod_equiv} and \ref{decoup_mod_equiv} apply to this problem. The optimal controls \req{u_opt_form}-\req{diff_r_i} in this case are again a combination of feedback and feedforward terms. The feedback gains are exactly the infinite-horizon LQR feedbacks, whereas the feedforward terms are  $\sT$-periodic functions given by the steady-state responses of the systems~\req{diff_g} to the $\sT$-periodic inputs $d_i$. Specifically, $\ric=\alpha$ solves an algebraic Riccati equation, and the steady-state responses of $\ff_i$ and $r_i$ satisfy
\begin{align}
	\dot{\ff}_i(t) 	~&=~ 	\tfrac{1}{\alpha} \ff_i(t) ~+~ d_i(t) , 						\label{d_to_ff.eq}			\\
	\dot{r}_i(t) 	~&=~ 	-\tfrac{1}{\alpha}  r_i(t)~-~ \tfrac{1}{\alpha^2}  \ff_i(t) .		\label{ff_to_r.eq}
\end{align}
Note that the response from $d_i$ to $r_i$ is the cascade of an anti-casual~\req{d_to_ff} followed by a causal~\req{ff_to_r} first-order linear time-invariant (LTI) system. The frequency response from $d_i$ to $r_i$ can therefore be found as the ratio  $\rha_i/\dha_i$ of their Fourier transforms, which is readily computed from~\req{d_to_ff}-\req{ff_to_r} as 
\begin{equation} \label{periodic_traj}
	\frac{\rha_i(\omega)}{\dha_i(\omega)}  
	~=~ \dfrac{1}{\alpha^2\omega^2 + 1}.
\end{equation}
Here, $\hat{r}_i$ and $\hat{d}_i$ are the Fourier coefficients of Dirac delta ``combs'' which are supported at harmonics of the fundamental frequency $\bar{\omega} = 2\pi/\sT$.
Observe that \eqref{periodic_traj} is the frequency response of a (non-causal) second-order low-pass filter with cutoff frequency $1 / {\alpha}$. In other words, the optimal trajectory tracks predominantly the low-frequency components of the reference. Also note that the phase response is identically zero, so that the state is perfectly in-phase with the reference. 
This is a reflection of the non-causal nature of the feedforward term. 

We can also write an expression for the cost \eqref{inf_hor_T.eq} in terms of the Fourier transform of $d_i$ as follows. Letting $\sT$ be the period of the reference signal and using Parseval's identity, we have
\begin{align}
	&J_i = \frac{1}{\sT} \smint{t}{t+\sT} (r_i(t) - d_i(t))^2 + \alpha^2 u_i^2(t) \, dt \\
	&= \left\| r_i - d_i \right\|_{L^2}^2 + \alpha^2 \left\| u_i \right\|_{L^2}^2
	=~\left\| \hat{r}_i - \hat{d}_i \right\|_{\ell^2}^2 + \alpha^2 \left\| \hat{u}_i \right\|_{\ell^2}^2 \nonumber \\
	&= \left\| \tfrac{-\alpha^2 \omega^2}{\alpha^2 \omega^2 + 1} \, \hat{d}_i \right\|_{\ell^2}^2 + \alpha^2 \left\| \tfrac{j \omega}{\alpha^2 \omega^2 + 1} \,\hat{d}_i \right\|_{\ell^2}^2
	= \left\| \tfrac{\alpha \omega}{\sqrt{\alpha^2 \omega^2 + 1}} \, \hat{d}_i \right\|_{\ell^2}^2 . \nonumber
\end{align}
The overall cost $J$ is therefore
\begin{equation}
	\smint{0}{1} \left\| \tfrac{\alpha \omega}{\sqrt{\alpha^2 \omega^2 + 1}}  \, \QDbh(z,\omega) \right\|_{\ell^2}^2 \, dz ~+~ K ,
\end{equation}
with $K$ now being expressed in the frequency domain as
\begin{equation}
	K ~=~ \smint{0}{1} \left\| \QDbh(z,\omega) - \hat{Q}_{\rmd}(z,\omega) \right\|_{\ell^2}^2 \, dz .
\end{equation}
We summarize this all in the following statement.

\begin{prop} \label{periodic_soln}
	Consider the problem \eqref{opt2.eq} with $\Omega \subset \R$, $D$ a $\sT$-periodic function, and 
	with the infinite-horizon performance objective~\req{inf_hor_T}. 
	The relation between the (temporal) Fourier transforms of the optimal  steady-state resource and demand 
	quantiles is given by 
	\be
	\QRh(z,\omega) ~=~\tfrac{1}{\alpha^2\omega^2+1} \, \QDbh (z,\omega), 
	\ee
	where the average $\QDb$ is taken with respect to the $\QR(\cdot,0)$-level-set partition of $[0,1]$. This solution attains the cost
	\begin{multline}
		J(R_0,D;\alpha) ~=~ \smint{0}{1} \left\| \tfrac{\alpha \omega}{\sqrt{\alpha^2 \omega^2 + 1}}  \, \QDbh(z,\omega) \right\|_{\ell^2}^2 \, dz \\
		~+~ \smint{0}{1} \left\| \QDbh(z,\omega) - \hat{Q}_{\rmd}(z,\omega) \right\|_{\ell^2}^2 \, dz .
	\end{multline}
\end{prop}

Recall that in the standard infinite-horizon LQR tracking problem the initial state plays no role in the steady-state response (i.e. the initial state is ``forgotten''). In the present problem, however, the initial state does play a role in the averaging operation $\QDb$. This is due to the additional (compared to the standard LQR tracking problem) constraint~\req{trans_mod_three}, which in turn affects the actual demand signals~\req{di_averaged} in the individual scalar LQ tracking problems.

Figures \ref{per_swarm} and \ref{per_sig} show a numerical example with a discrete resource and a continuous demand given by a periodically time-varying Gaussian mixture
\begin{multline}
	D(x,t) ~=~ (1 + \sin (2 \pi t)) \, \mathcal{N}(2.5,1) \\
	~+~ (1 - \sin (2 \pi t)) \, \mathcal{N}(7.5,1) .
	\label{dem_periodic.eq}
\end{multline}

\begin{figure}[t]
	\centering
	\includegraphics[width=0.9\linewidth]{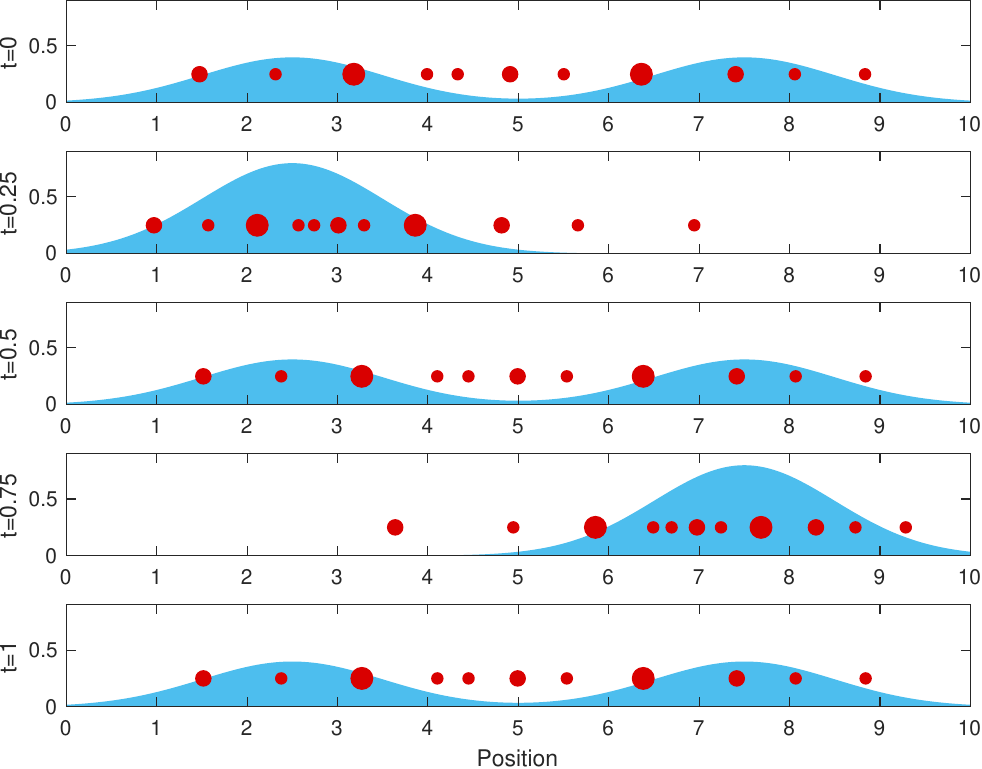}
	\mycaption{Timeseries of steady-state resource (red) and demand (blue) for $\alpha = 0.08$. The resource consists of eleven agents of unequal mass, while the demand consists of a periodically-time-varying Gaussian mixture. Notice that the resource distribution does not track the demand distribution as closely as possible due to the non-negligible motion cost.}
	\label{per_swarm}
\end{figure}

\begin{figure}[t]
	\centering
	\includegraphics[width=0.88\linewidth]{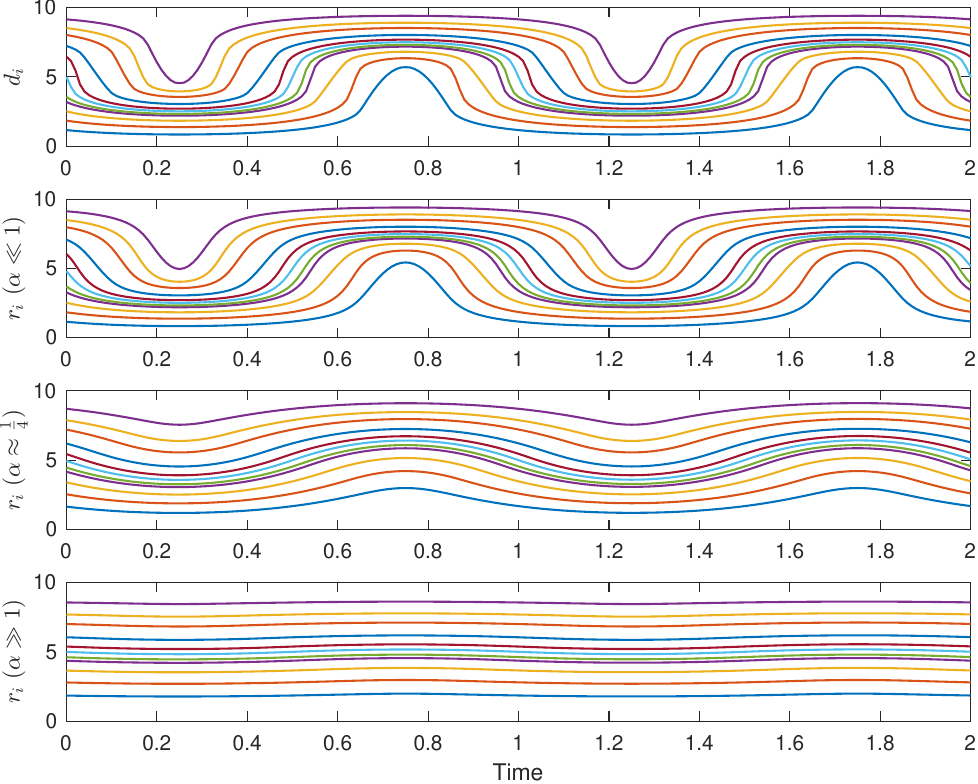}
	\mycaption{The tracking signals $d_i$ (top) and resulting steady-state trajectories $r_i$ for small, moderate, and large values of $\alpha$. As $\alpha$ becomes larger, the resource swarm moves less due to the higher motion cost. The optimal resource trajectory can be interpreted as low-pass-filtered version of the demand trajectory with cutoff frequency $1 / \alpha$. Notice that while the demand distribution~\req{dem_periodic} is a pure sinusoid at each $x$, the transformation from $D$ to $\QD$ is nonlinear and thus the corresponding quantile functions contain higher harmonics.}
	\label{per_sig}
\end{figure}


\section{Conclusion} \label{conclusion}

In this paper, we introduced a novel model for tracking control for large swarms of autonomous agents. We showed that when the spatial domain is one-dimensional, this problem can be transformed into an equivalent family of decoupled scalar LQ tracking problems. Full solutions were provided in this case, and the special cases where the reference is static and periodically time-varying were further investigated. We note that explicit solutions like those presented in this paper are unusual for these types of problems. We thus find these results to be promising, and present several immediately obvious directions for future work. 
1) There are interesting cases where the mass of distributions may be changing or where some resources or demands may go unassigned, and so it is desirable to generalize assumption \eqref{normalized_mass_eq}. This may be accommodated using various notions of ``unbalanced optimal transport'' \cite{Chizat2018,Gangbo2019}. 
2) While the necessary conditions \eqref{necc_cond_1}-\eqref{necc_cond_2} appear to be very challenging at first, we believe that they are fairly structured. This gives hope that reasonable numerical schemes may actually be possible.
3) While the solutions presented here were centralized, they do have some inherent ``distributed structure''. In particular, resource agents only need knowledge of the location of their assigned demands. Thus, by computing these assignments offline or in a distributed manner (as in \cite{Krishnan2018a}), one can obtain a naturally distributed implementation. 
4) It is desirable to translate the optimal noncausal solutions presented here into causal solutions which could be used for real-time tracking for a priori-unknown reference signals. 
5) While the one-dimensional case presented here could be used for swarm deployment in one-dimensional settings (e.g., along a highway or border), the two-and three-dimensional cases are clearly of high interest. Indeed, the transformation to quantile functions has a natural generalization in higher dimensions, and in certain cases one can obtain a similar decoupling into LQ problems. 
Some preliminary results in directions (4) and (5) can be found in our recent papers \cite{Emerick2023,Emerick2024}.


\appendix

We recall the following facts about the CDF $F_\mu$ and quantile function $Q_\mu$ which we will use in our proofs \cite{Kaempke2015}:
\begin{align}
	Q_\mu \of F_\mu (x) &~=~ \min \{ \xi : F_\mu(\xi) = F_\mu(x) \} \label{q_of_f} \\ 
	Q_\mu \of F_ \mu &~=~ I \qquad \mu \text{-a.e.} \label{q_f_inverses} \\ 
	F_\mu \of Q_\mu(z) &~=~ \max \{ \zeta : Q_\mu(\zeta) = Q_\mu(z) \} . \label{f_of_q}
\end{align}

\subsection{Proof of Lemma~\ref{equivalent_dynamics}} \label{equivalent_dynamics_appen}

\begin{proof}
	Equation~\req{F_PF} is related to~\req{R_PF} by integration
	\begin{align}
		\smint{\infty}{x}	\!\! \partial_t  R(\xi,t) ~d\xi 
		&~=~	 - \smint{\infty}{x} \!\! \partial_\xi \big( V(\xi,t) ~R(\xi,t) 	\rom \big)  d\xi 		\nonumber	\\
		\Rightarrow\hsom
		\partial_t  \smint{\infty}{x}	\!\!  R(\xi,t) ~d\xi 
		&~=~	 -  V(x,t) ~R(x,t) 				\\
		\Rightarrow\hsom 		 
		\partial_t \FR(x,t) 	&~=~ 	-  V(x,t) ~\partial_x \FR(x,t) . \nonumber
	\end{align}
	To relate~\req{Q_PF} to \eqref{F_PF.eq} we first recall that the solutions of~\req{F_PF} can be expressed as follows~\cite{Olver2014}.  
	Let $\Phi$ be the flow map of the vector field $V$
	\be
	\partial_t  \Phi_t(x)  ~=~ V  \big( \Phi_t(x), t		\rom  \big) , 
	\hstm 
	\Phi_0(x)~=~x. 
	\label{flow_map.eq}
	\ee
	Then the solution $\FR$ at any time is given by the {\em function pushforward} of $\FR$ by the flow $\Phi$
	\be
	\begin{aligned} 
		\forall x &\in\R, & 
		\FR(x,t) ~&=~ \FR \!\lp \Phi_t^{\sm1}(x),0 \rp  ,					\\ 
		\Leftrightarrow\hstm 
		\forall x &\in\R, & 
		\FR \! \lp \Phi_t(x),t \rp ~&=~ \FR \!\lp x,0 \rp .				
	\end{aligned} 
	\label{F_0_t.eq}
	\ee
	We can use the last equation to introduce a change of variables 
	\be
	z ~:=~ \FR(x,0) ~=~ \FR \! \lp \Phi_t(x),t \rp .
	\label{z_change.eq}
	\ee
	Acting with $Q(\cdot,t)$ on these last equations gives 
	\begin{equation} 
		\QR(z,t) = \QR \!\lp \FR(x,0), t\rp  = 	\QR \!\lp \FR \lp \Phi_t(x),t \rp,t \romn \rp 			
		= \Phi_t(x) 				,
	\end{equation} 
	where the last equality follows from \eqref{q_f_inverses}. 
	Taking time derivatives of both sides and using the flow map property~\req{flow_map}
	\begin{equation}
		\partial_t \QR(z,t) = \partial_t \Phi_t(x) = V  \!\lp \Phi_t(x), t \romn  \rp 
		=  V  \!\lp \QR(z,t), t \romn  \rp, 
	\end{equation}
	which is exactly equation~\req{Q_PF}.
	Finally, to relate \eqref{R_PF.eq} to \eqref{Q_PF.eq}, we first act with $Q(\cdot,0)$ on \eqref{z_change.eq} to obtain $x = \QR(z,0)$. We can then write $\QR(z,t) ~=~ \Phi(\QR(z,0),t)$. Applying Lemma \ref{Q_pushfw_lem} and properties of the pushforward, we can write
	\begin{multline} \label{QR_pushfw}
		R_t ~=~ \left[ \QR(\cdot,t) \right]_\# \bone~=~  \left[ \Phi_t(\QR(\cdot,0)) \right]_\# \bone \\
		~=~ \left[ \Phi_t \right]_\#  \big( \left[ \QR(\cdot,0))\right]_\# \bone  \big) ~=~  \left[ \Phi_t \right]_\# R_0 .
	\end{multline}
	We then know from \cite[Theorem~4.4]{Santambrogio2015} that this $R$ along with the $V$ which generates $\Phi$ satisfy \eqref{R_PF.eq}.	
\end{proof}

\subsection{Proof of Lemma \ref{equivalent_constraints_lem}} \label{equivalent_constraints_appen}

\begin{proof}
	The central fact that we need to prove is that $\mathcal{T}$ is a bijection between the two sets. It suffices to show then that (1) $\mathcal{T}(R,V)$ is in the second set, (2) that $\mathcal{T}^{-1}(\QR,U)$ is in the first set, and (3) that $\mathcal{T}$ and $\mathcal{T}^{-1}$ are actually inverses.
	
	(1) We have $\mathcal{T}(R,V) = (\QR,U) = (\QR, V \of \QR)$. We can see that $(\QR,U)$ satisfy the dynamic constraint by Proposition \ref{equivalent_dynamics}, $\QR(\cdot,0)$ satisfies the monotonicity constraint because quantile functions are monotone nondecreasing, and $U$ satisfies the input constraint because of the composition with $\QR$.
	
	(2) We have that $\mathcal{T}^{-1}(\QR,U) = (R,V) = (R, U \of \FR)$. We can see that as long as $\QR(\cdot,t)$ is monotone nondecreasing, it is an actual quantile function and so $R$ and $\FR$ are well-defined. Proposition \ref{equivalent_dynamics} then tells us that $(R,V)$ satisfy the dynamic constraint. We claim that the monotonicity of $\QR(\cdot,t)$ follows from the second and third constraints. This is a straightforward proof by contradiction: suppose that $\QR(z_1,t)$ were $> \QR(z_2,t)$ for some $z_1 < z_2$. Then since $\QR(z_1,0) \leq \QR(z_2,0)$, by the intermediate value theorem, there must exist some time $t^* \in [0,t)$ where $\QR(z_1,t^*) = \QR(z_2,t^*)$. However, this would require that $U(z_1,t') = U(z_2,t')$ and $\QR(z_1,t') = \QR(z_2,t')$ for all $t' \geq t^*$. This is a contradiction, and so we know that $\QR(\cdot,t)$ must be monotone.
	
	(3) We can see that $\mathcal{T}^{-1} \of \mathcal{T} = I$ as long as $V \of \QR \of \FR = V$. This is ensured because $\QR \of \FR = I$ $\mu$-a.e. \eqref{q_f_inverses}, and we only care about the value of $V$ up to $\mu$-a.e. equivalence. Similarly, we can see that $T \of T^{-1} = I$ as long as $U \of \FR \of \QR = U$. This is ensured by the input constraint because $\FR \of \QR (z) = \max \{ \zeta : \QR(\zeta) = \QR(z) \}$ \eqref{f_of_q}.
\end{proof}

\subsection{Proof of Proposition \ref{decoup_mod_equiv}} \label{decoup_mod_equiv_appen}

\begin{proof}
	To complete the proof, we wish to show that the constraint \eqref{trans_mod_three.eq} is automatically satisfied by $(\QR,U)$ as reconstructed from the solutions to the (independent) scalar LQ tracking problems via \eqref{reconstruction}. It suffices to show that solutions to these scalar problems satisfy
	\begin{equation} \label{subset_constraint}
		r_i(t) ~=~ r_j(t) \quad \Rightarrow \quad u_i(t) ~=~ u_j(t) .
	\end{equation}
	This condition is satisfied trivially if $i = j$ and so we only need to consider the case where $i \neq j$. In this case, it turns out that $r_i(t)$ is never equal to $r_j(t)$, and so the condition holds.
	
	To prove this, suppose without loss of generality that $i < j$. Then $r_i(0) < r_j(0)$. Since $\QR(\cdot,0)$ and $\QD(\cdot,t)$ are both monotone nondecreasing, $d_i(t) \leq d_j(t)$ for all $t$. It is then seen from equation \eqref{g_integral.eq} that since $\phi_\rmy$ is always positive, $\ff_i(t) \geq \ff_j(t)$ for all $t$. Now, define $\beta(t) := r_j(t) - r_i(t)$ to be the difference between the respective solutions, differentiate, and apply the dynamics $\dot{r}_i = u_i$ with the optimal control \eqref{input_u} to obtain
	\begin{equation}
		\dot{\beta} ~=~ - \tfrac{1}{\alpha^2}  \ric \beta ~-~ \tfrac{1}{\alpha^2} ( \ff_j - \ff_i ) .
	\end{equation}
	Using the fact that $\ff_i(t) \geq \ff_j(t)$, we have that $\dot{\beta} \geq - \tfrac{1}{\alpha^2}  \ric \beta$. By comparison, then, and since $\beta(0) > 0$, we have that
	\begin{equation}
		\beta(t) ~\geq~ \beta(0) \, \exp\lB\sm  \smint{\tau}{t}  \tfrac{1}{\alpha^2}  \ric (s) \,ds \rB ,
	\end{equation}
	and we conclude that $\beta(t) > 0$ or that $r_i(t) < r_j(t)$.
\end{proof}




\vspace{-15mm}

\begin{IEEEbiography}[{\includegraphics[width=1in,height=1.25in,clip,keepaspectratio]{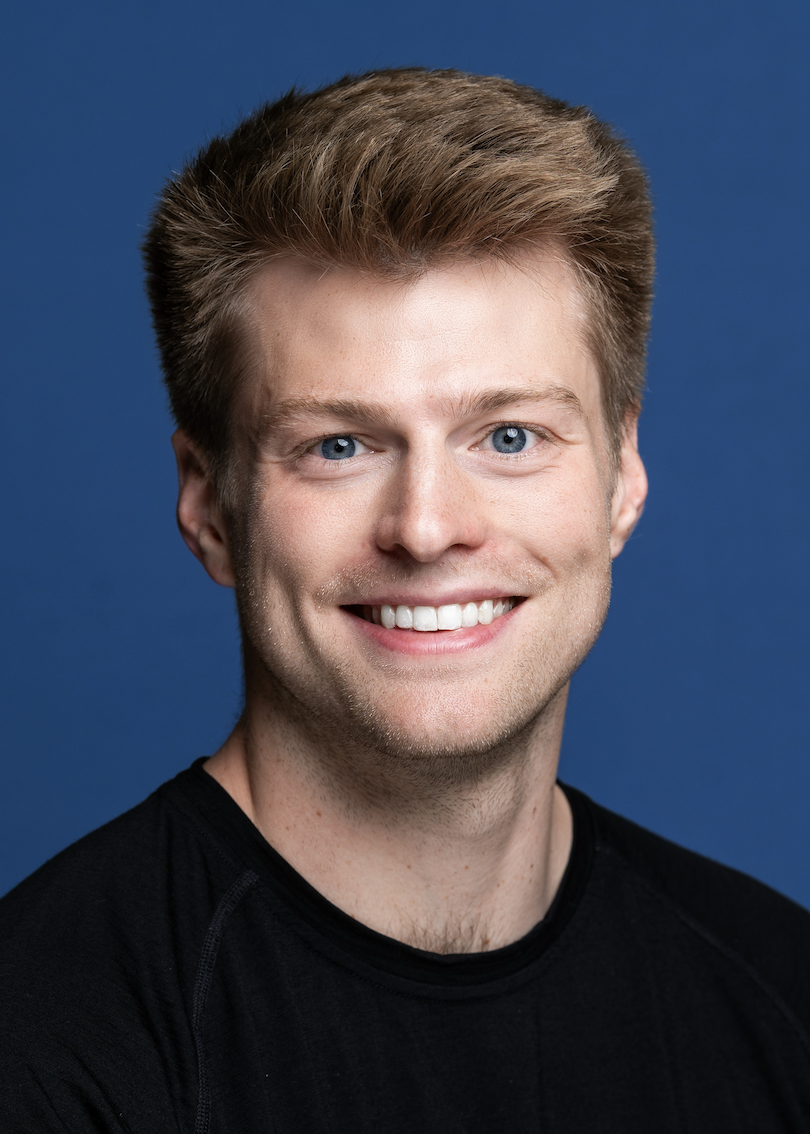}}]{Max Emerick} (GS'22)
	received the B.S. degree in mechanical engineering from California Polytechnic State University San Luis Obispo in 2020 and the M.S. degree in mechanical engineering from the University of California Santa Barbara in 2022. He is currently pursuing the Ph.D. degree in mechanical engineering at the University of California Santa Barbara.
	His research interests include optimal control and distributed parameter systems.
	Mr. Emerick is a recipient of the IFAC Young Author Award, the Conference on Decision and Control Outstanding Student Paper Award, and the National Defense Science and Engineering Graduate (NDSEG) Fellowship.
\end{IEEEbiography}

\vspace{100mm}

\begin{IEEEbiography}[{\includegraphics[width=1in,height=1.25in,clip,keepaspectratio]{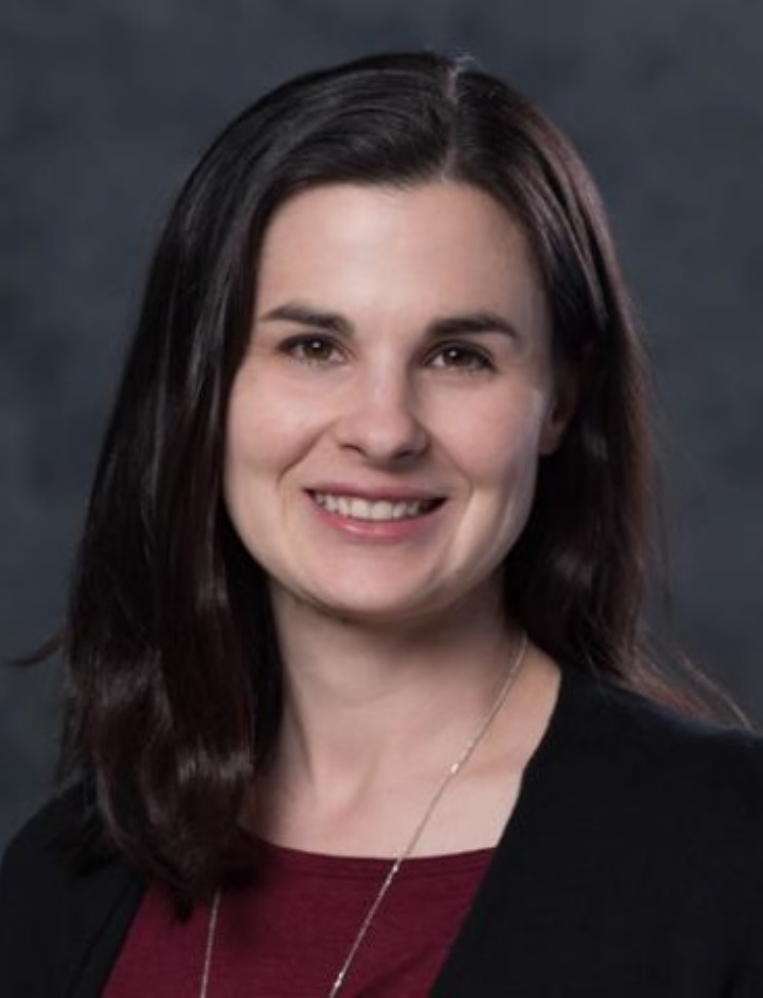}}]{Stacy Patterson} 
	Stacy Patterson is an Associate Professor in the Department of Computer Science at Rensselaer Polytechnic Institute. She received the M.S. and Ph.D. in computer science from UC Santa Barbara in 2003 and 2009, respectively. From 2009-2011, she was a postdoctoral scholar at the Center for Control, Dynamical Systems and Computation at UC Santa Barbara. From 2011-2013, she was a postdoctoral fellow in the Department of Electrical Engineering at Technion - Israel Institute of Technology. Dr. Patterson is the recipient of a Viterbi postdoctoral fellowship, the IEEE CSS Axelby Outstanding Paper Award, and an NSF CAREER award. Her research interests include distributed systems and machine learning.
\end{IEEEbiography}

\vspace{-115mm}

\begin{IEEEbiography}[{\includegraphics[width=1in,height=1.25in,clip,keepaspectratio]{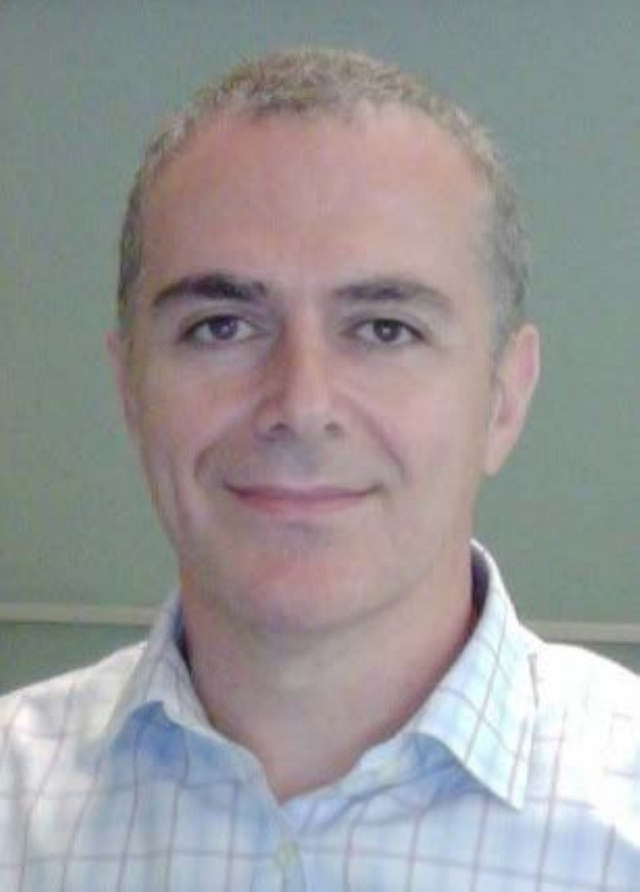}}]{Bassam Bamieh}
(Fellow, IEEE) received the B.Sc. degree in electrical engineering and physics from Valparaiso University, Valparaiso, IN, USA, in 1983, and the M.Sc. and Ph.D. degrees in electrical and computer engineering from Rice University, Houston, TX, USA, in 1986 and 1992, respectively.
From 1991 to 1998, he was an Assistant Professor with the Department of Electrical and Computer Engineering, and the Coordinated Science Laboratory, University of Illinois
at Urbana-Champaign, after which he joined the University of California at Santa Barbara (UCSB), where he is currently a Professor of Mechanical Engineering. His research interests include robust and optimal control, distributed and networked control and dynamical systems, and related problem in fluid and statistical mechanics and thermoacoustics.
Dr. Bamieh is a past recipient of the IEEE Control Systems Society G. S. Axelby Outstanding Paper Award (twice), the AACC Hugo Schuck Best Paper Award, and the National Science Foundation CAREER Award. He was a Distinguished Lecturer of the IEEE Control Systems Society (twice), and is a Fellow of the International Federation of Automatic Control (IFAC).                                                                                     
\end{IEEEbiography}

\end{document}